\documentclass[a4paper,UKenglish]{lipics}
 
\usepackage{microtype}
\usepackage{amsmath,amsthm,amsfonts,dsfont}
\usepackage{amssymb,latexsym,graphicx}
\usepackage{enumitem}
\usepackage{stmaryrd,eucal}
\usepackage{mathtools}
\usepackage{hyperref}

\usepackage[lined,boxed]{algorithm2e}
\SetKwFor{Loop}{loop}{}{end}
\SetKwInput{KwReq}{Precondition}

\usepackage{tikz}
\usetikzlibrary{calc}
\usetikzlibrary{intersections}

\renewcommand{\vec}[1]{\mathbf{#1}}
\newcommand{\pr}[2]{\left\langle{#1, #2}\right\rangle}

\newcommand{\set}[1]{\left\{#1\right\}}

\newcommand{\R}{\ensuremath{\mathbb{R}}}
\newcommand{\Z}{\ensuremath{\mathbb{Z}}}

\newcommand{\eqdef}{\mathbin{\stackrel{\rm def}{=}}}

\def\vol{{\rm vol}}
\def\conv{{\rm conv}}

\def\eps{\varepsilon}
\def\d{{\rm d}}

\def\linsp{{\rm span}}
\def\cone{{\rm cone}}
\def\aff{{\rm aff}}

\makeatletter
\def\imod#1{\allowbreak\mkern8mu({\operator@font mod}\,\,#1)}
\makeatother

\DeclareMathOperator{\E}{\mathbb{E}}

\DeclareMathOperator*{\argmax}{arg\,max}

\newcommand{\poly}{\ensuremath{{\rm poly}}}

\newcommand {\cB}         {\mathcal{B}}
\newcommand {\cP}         {\mathcal{P}}
\newcommand {\cT}         {\mathcal{T}}

\newcommand {\bS}         {\mathbb{S}}

\newcommand {\cU}         {\mathcal{U}}

\newtheorem{question}[theorem]{Question}

\newif\ifnotes\notesfalse
\ifnotes
\usepackage{color}
\definecolor{mygrey}{gray}{0.50}
\newcommand{\notename}[2]{{\textcolor{mygrey}{\footnotesize{\bf (#1:} {#2}{\bf ) }}}}

\else
\newcommand{\notename}[2]{{}}

\fi

\title{On the Shadow Simplex method for curved polyhedra}
\titlerunning{On the Shadow Simplex method for curved polyhedra} 

\author[1]{Daniel Dadush}
\author[2]{Nicolai H{\"a}hnle}
\affil[1]{Centrum Wiskunde \& Informatica, Netherlands \\
\texttt{dadush@cwi.nl}}
\affil[2]{Universit{\"a}t Bonn, Germany \\ 
  \texttt{haehnle@or.uni-bonn.de}}
\authorrunning{D. Dadush and N. H{\"a}hnle} 

\Copyright{Daniel Dadush and Nicolai H{\"a}hnle}


\serieslogo{}
\volumeinfo
  {Billy Editor and Bill Editors}
  {2}
  {Conference title on which this volume is based on}
  {1}
  {1}
  {1}
\EventShortName{}
\DOI{10.4230/LIPIcs.xxx.yyy.p}

\begin{document}

\maketitle

\begin{abstract}
We study the simplex method over polyhedra satisfying certain ``discrete
curvature'' lower bounds, which enforce that the boundary always meets vertices
at sharp angles. Motivated by linear programs with totally unimodular constraint
matrices, recent results of Bonifas et al (SOCG 2012), Brunsch and Röglin
(ICALP 2013), and Eisenbrand and Vempala (2014) have improved our understanding
of such polyhedra.

We develop a new type of \emph{dual} analysis of the shadow simplex method which
provides a clean and powerful tool for improving all previously mentioned
results. Our methods are inspired by the recent work of Bonifas and the first
named author~\cite{BD14}, who analyzed a remarkably similar process as part of
an algorithm for the Closest Vector Problem with Preprocessing.  

For our first result, we obtain a constructive diameter bound of
$O(\frac{n^2}{\delta} \ln \frac{n}{\delta})$ for $n$-dimensional polyhedra with
curvature parameter $\delta \in [0,1]$. For the class of polyhedra arising from
totally unimodular constraint matrices, this implies a bound of $O(n^3 \ln n)$.
For linear optimization, given an initial feasible vertex, we show that an
optimal vertex can be found using an expected $O(\frac{n^3}{\delta} \ln
\frac{n}{\delta})$ simplex pivots, each requiring $O(m n)$ time to compute. An
initial feasible solution can be found using $O(\frac{m n^3}{\delta} \ln
\frac{n}{\delta})$ pivot steps.
\end{abstract}

\section{Introduction}

The \emph{simplex method} is one of the most important methods for solving linear programs (LP),
that is, optimization problems of the form $\max \set{\pr{\vec{c}}{\vec{x}}: \vec{x} \in P}$
where $P$ is a polyhedron defined by linear constraints.
Starting from an initial vertex $\vec{v}$,
a simplex algorithm provides a rule for moving from vertex to vertex along edges
of the graph or $1$-skeleton of $P$ until an
optimal vertex $\vec{w}$ (or an unbounded ray) is found.

A long standing open question is whether there exists a polynomial-time simplex algorithm for LP.
The first obstacle in proving the existence (or
non-existence) of such a method is the following fundamental question:

\begin{question} Given any two vertices $\vec{v},\vec{w}$ of a polyhedron $P$,
what is the best possible bound on the length of the shortest path between them,
as a function of the dimension $n$ and the number of constraints $m$?
\end{question}

The polynomial Hirsch conjecture posits that the diameter of the graph of a
polyhedron is bounded by a polynomial in $m$ and $n$. The best known general upper
bounds are however much larger. Barnette~\cite{Barnette74} and
Larman~\cite{Larman70} proved a bound of $O(2^n m)$, and Todd~\cite{Todd14}
recently proved a bound of $(m - n)^{\log n}$, slightly improving an earlier
bound of Kalai and Kleitman~\cite{Kalai91,KK92}. The original Hirsch conjecture,
which posited a bound of $m-n$, was
recently disproved for polytopes (i.e.~ bounded polyhedra) by Santos~\cite{Santos12,MSW13}, who
gave a lower bound of $(1+\eps)m$ (only slightly violating the conjectured bound).

Given the difficulty of the general question, much research has been aimed at
bounding the diameter of special classes of polyhedra. For example, polynomial
bounds have been given for $0/1$ polytopes~\cite{Naddef89}, transportation
polytopes~\cite{B84,BHS06,LKOS09}, and flag polytopes~\cite{AB14}.

Another important class, which has recently received much attention and is
directly related to this work, are polyhedra whose constraint matrices are
``well-conditioned''. Dyer and Frieze~\cite{DF94} showed that the diameter of totally
unimodular polyhedra -- i.e. having integer constraint matrices with all
subdeterminants in $\set{0,\pm 1}$ -- is bounded by $O(m^{16}n^3(\log
mn)^3))$. Their work also contains a polynomial time
randomized simplex algorithm that solves linear programs over totally unimodular
polyhedra.

The diameter bound of Dyer and Frieze was both generalized and improved in the
work of Bonifas et al~\cite{BSEHN14}. They showed that polyhedra with integer
constraint matrices and all subdeterminants bounded by $\Delta$ have diameter
$O(\Delta^2 n^4 \log(n \Delta))$ if they are unbounded and $O(\Delta^2 n^{3.5}
\log(n \Delta))$ if they are bounded. Their proof used certain expansion
properties of the polyhedral graph and was non-constructive.

In an attempt to make the bound of~\cite{BSEHN14} constructive, Brunsch and
Röglin~\cite{BR13} showed that given any two vertices $\vec{v},\vec{w}$ on such a
polyhedron $P$, a path between them of length $O(m \Delta^4 n^4)$ (note the dependence on
$m$) can be constructed using the \emph{shadow simplex method}.
In fact, they give a more general bound based on the so-called
$\delta$-distance property of the constraint matrix, which measures how ``well
spread'' the rows of the constraint matrix are\footnote{We note that this measure is
already implicit in~\cite{BSEHN14} and that the diameter bound factors through it.}.
Using this parameter they give a bound of $O(m n^2/\delta^2)$ on the length of
the constructed path, and recover the previous bound by the relationship $\delta
\geq 1/(n \Delta^2)$.

Most recently, Eisenbrand and Vempala~\cite{EV14} provided a different approach to
making the Bonifas et al~\cite{BSEHN14} result constructive, which more closely
resembles the random walk approach of Dyer and Frieze and also extends to
optimization. When the constraint matrix satisfies the $\delta$-distance
property, they show that given an initial vertex and objective, an optimal
vertex can be computed using $\poly(n,1/\delta)$ random walk steps (no
dependence on $m$). Furthermore, an initial feasible vertex can be computed
using $m$ calls to their optimization algorithm over subsets of the original
constraints.

\section{Results}
\label{sec:results}

Building and improving upon the works of Bonifas et al~\cite{BSEHN14}, Brunsch
and R{\"o}glin~\cite{BR13} and Vempala and Eisenbrand~\cite{EV14}, we give an
improved (constructive) diameter bound and simplex algorithm for polyhedra
satisfying the $\delta$-distance and other related properties. We also make
improvements in the treatment of unbounded polyhedra and degeneracy.  All our
results are based on a new variant and analysis of the \emph{shadow simplex
method}.  

We now introduce the ``discrete curvature measures'' we use along with the
corresponding results. We list these measures in order of increasing strength.
In the next section, we shall explain our variant of the shadow simplex method
and compare it with previous implementations. 

Let $P = \set{\vec{x} \in \R^n: A\vec{x} \leq \vec{b}}$, $A \in \R^{m \times
n}$, $\vec{b} \in \R^n$ be a pointed polyhedron ($A$ has full column rank
$\Leftrightarrow$ $P$ has vertices). For a vertex $\vec{v}$ of $P$, the \emph{normal
cone} at $\vec{v}$ is $N_\vec{v} = \set{\sum_{i \in I_{\vec{v}}} \lambda_i
\vec{a}_i: \lambda_i \geq 0, i \in I_{\vec{v}}}$, where $I_{\vec{v}} = \set{i
\in [m]: \pr{\vec{a}_i}{\vec{v}} = \vec{b}_i}$ is the set of tight constraints.
Equivalently, $N_{\vec{v}}$ is the set of all linear objective
functions whose maximum over $P$ is attained at $\vec{v}$.  $N_{\vec{v}}$ is
simplicial (non-degenerate) if it is generated by a basis of $A$, that is, if
exactly $n$ linearly independent constraints of $P$ are tight at $\vec{v}$. The
\emph{normal fan} of $P$ is the collection of all the vertex normal cones,
and the \emph{support of the normal fan} $N(P)$ is their union. A polyhedron is simple (or
non-degenerate) if all its vertex normal cones are simplicial.

\begin{definition}[$\tau$-wide Polyhedra]
We say that a cone $C$ is $\tau$-wide if it contains a
Euclidean ball of radius $\tau$ centered on the unit sphere. We define a
polyhedron $P$ to have a $\tau$-wide normal fan (or simply $P$ to be
$\tau$-wide) if every vertex normal cone is $\tau$-wide. 
\end{definition}

    
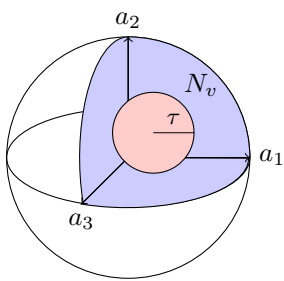
\begin{center}
      \begin{tikzpicture}[scale=0.8]
        \draw (0,0) circle (2cm);
        \draw (0,0) ellipse (2cm and .82cm);
 
        \filldraw[fill=blue!20] (0,0) -- (0,2,0) --
                                  (0,2,0) arc (90:202:.8 and 2) -- cycle; 
        \filldraw[fill=blue!20] (0,0) -- (2,0,0) --
                                  (2,0,0) arc (0:90:2) -- cycle; 
        \filldraw[fill=blue!20] (0,0) -- (0,0,2) --
                                (0,0,2) arc (248:360:2cm and .82cm) -- cycle; 
         \draw[->] (0,0) -- (2,0,0) node [right] {$a_1$};
         \draw[->] (0,0) -- (0,2,0) node [above] {$a_2$};
         \draw[->] (0,0) -- (0,0,2) node [below] {$a_3$};
         \filldraw[fill=red!20] (.67,.67,.67) circle (.67);
                   \draw (1.2,1.2) node {$N_v$};
                   \draw[very thin] (.67,.67,.67) -- node[above] {\small $\tau$} +(.67,0);
      \end{tikzpicture}
    \end{center}

Roughly speaking, having a $\tau$-wide normal fan enforces that facets always
intersect at ``sharp angles'' (i.e.~angle bounded away from $\pi$). In
particular, for any vertex $\vec{v}$ of $P$, the angle between any two rays
emanating from $\vec{v}$ and (non-trivially) passing through $P$ is at most
$\pi-2\tau$. Hence one can interpret this condition as a discrete form of
curvature for polyhedra.  We now state our diameter bound for $\tau$-wide
polyhedra.

\begin{theorem}[Diameter Bound, see Theorem~\ref{thm:diameter-bound}]
\label{thm:diameter-bound-results}
Let $P \subseteq \R^n$ be an $n$-dimensional pointed polyhedron having a
$\tau$-wide normal fan. Then the graph of $P$ has diameter bounded by
$8n/\tau(1+\ln(1/\tau))$. Furthermore, a path of this expected length can be
constructed via the shadow simplex method. 
\end{theorem}

Restricting to $n$-dimensional polyhedra with subdeterminants bounded by
$\Delta$, using the relation $\tau \geq 1/(n\Delta)^2$ (see Lemma~\ref{lem:determinant-delta-dist}) we achieve
a bound of $O(n^3\Delta^2\ln(n\Delta))$, improving on the existential bounds of
Bonifas et al~\cite{BSEHN14}. In contrast to~\cite{BSEHN14}, we note that our
bound (and proof) is the same for polytopes and unbounded polyhedra.   

While our bound is constructive -- we follow a shadow
simplex path -- it is in general only efficiently implementable when the
polyhedron is simple. In the presence of degeneracy, we note that computing a
single edge of the path is essentially as hard as solving linear programming.
Furthermore, standard techniques for removing degeneracy, such as the
perturbation or lexicographic method, may unfortunately introduce a large number
of extra simplex pivots. 

Interestingly, our diameter bound can take advantage of degeneracy in
situations where it makes the normal cones wider. While degeneracy does not
occur for ``generic polyhedra'', it is very common for combinatorial polytopes.
Furthermore, it can occur in ways that are useful to our diameter bound. For
example, we remark that using degeneracy one can prove that the normal
fan of the perfect matching polytope is $\Omega(1/\sqrt{|E|})$-wide,
see Appendix~\ref{sec:perfect-matching}.


To solve linear optimization problems via the shadow simplex method, we will
need more than a wide normal fan. In fact, we will have different requirements
for the two phases of the simplex algorithm: Phase $1$, which finds an initial
feasible vertex, will require more than Phase $2$, which finds an optimal vertex
with respect to the objective starting from a feasible vertex.

\begin{definition}[$\delta$-distance property]
A set of linearly independent vectors $\vec{v}_1,\dots,\vec{v}_k \in \R^n$
satisfy the $\delta$-distance property if for every $i \in [k]$, the vector
$\vec{v}_i$ is at Euclidean distance at least $\delta \|\vec{v}_i\|$ from the
span of $\set{\vec{v}_j: j \in [k] \setminus \set{i}}$. 

For a polyhedron $P = \set{\vec{x} \in \R^n: A\vec{x} \leq \vec{b}}$, we define
$P$ to satisfy the \emph{local $\delta$-distance property} if every
\emph{feasible basis} of $A$, i.e.~the rows of $A$ defining a vertex of $P$,
satisfies the $\delta$-distance property. 

We say that a set of vectors $\vec{v}_1,\ldots,\vec{v}_m \in \R^n$
satisfy the \emph{global $\delta$-distance property}
if every linearly independent subset satisfies the $\delta$-distance property.
We say that a matrix $A \in \R^{m \times n}$ satisfies the global $\delta$-distance property
if its row vectors do.
\end{definition}

\begin{lemma}
  \label{lem:dist-to-width}
  Let $\vec{v}_1,\dots,\vec{v}_n \in
  \bS^{n-1}$ be a basis satisfying the $\delta$-distance property. Then
  $\cone(\vec{v}_1,\dots,\vec{v}_n)$ is $\delta/n$-wide.
\end{lemma}
\begin{proof}
  See Appendix~\ref{sec:appendix-proofs}.
\end{proof}

The definitions differ in strength mainly based on the sets of bases to which they apply.
The local $\delta$-distance property is stronger than the $\tau$-wide property for $\tau = \delta/n$,
because it implies that \emph{all triangulations} of the normal fan are $\tau$-wide.%
\footnote{However, the $\tau$-wide property is weaker even when all normal cones are simplicial:
a $2$-dimensional cone of inner angle close to $\pi$ is almost $1$-wide,
but satisfies $\delta$-distance only for $\delta$ close to~$0$.}
The global property is stronger than the local property since it applies also to infeasible bases,
which allows one to control the geometry of polyhedra related to $P$,
such as polyhedra obtained by removing a subset of constraints, which will be needed for Phase $1$.


We now state our main result for Phase $2$ simplex.

\begin{theorem}[Optimization via Shadow Simplex, see Theorem~\ref{thm:optimization-detail}]
\label{thm:phase2}
Let $P = \set{\vec{x} \in \R^n: A\vec{x} \leq \vec{b}}$ be an $n$-dimensional
\emph{polytope} with $m$ constraints satisfying the \emph{local
$\delta$-distance property}. Then, given an objective $\vec{c} \in \R^n$ and a
vertex $\vec{v}$ of $P$, an optimal vertex can be computed using an
expected $O(n^3/\delta \ln(n/\delta))$ shadow simplex pivots, where each pivot
requires $O(mn)$ arithmetic operations.
\end{theorem}

Our Phase $2$ algorithm above is faster than the algorithms in~\cite{BR13,EV14}
and relies on a \emph{weaker assumption} than~\cite{EV14}. The
$\vec{v}$,$\vec{w}$ path finding algorithm of Brunsch and
R{\"o}glin~\cite{BR13} is in fact a special case of the above, since we can
choose $\vec{c}$ to be any objective maximized at $\vec{w}$. Comparing to the
Phase $2$ algorithm of Eisenbrand and Vempala~\cite{EV14}, we require only the
\emph{local $\delta$-distance property} instead of the global one. Whether one
could rely only on the local property was left as open question in~\cite{EV14},
which we resolve in the affirmative.     

A small technical caveat is that as stated, the algorithm requires knowledge of
$\delta$. Since $\delta \leq 1$, we can always guess a number $\delta' \leq
\delta \leq 2\delta'$ by trying $O(\ln 1/\delta)$ different values, incurring an
$O(\ln 1/\delta)$ factor increase in running time (overestimating $\delta$ only
affects correctness, not runtime). For simplicity, we shall henceforth assume
that $\delta$ is known.

A more important caveat is that the above algorithm requires that $P$ be a
polytope (i.e.~bounded). This restriction is due to the fact that we can only
generate the randomness required for our bounds efficiently (that is, without
solving a general LP) when the support of the normal fan equals $\R^n$.

The unbounded setting can be reduced to the bounded setting, in the standard way,
by adding one or more constraints to make $P$ bounded while not cutting off any
of its vertices.

\begin{definition}
  Let $P = \set{ \vec{x} \in \R^n : A\vec{x} \leq \vec{b} }$ be a pointed polyhedron.
  Then a polytope $P' = \set{ \vec{x} \in \R^n : A\vec{x} \leq \vec{b}, A' \vec{x} \leq \vec{b}' }$
  is \emph{LP equivalent} to $P$ if every vertex $\vec{v} \in P$ satisfies
$\pr{\vec{a}_i'}{\vec{v}} < \vec{b}_i'$ for all $i$;
  in particular, $x$ is a vertex of $P'$.
\end{definition}

Given an optimal vertex $\vec{v}$ of $P'$ as above, one can easily check whether
$\vec{v}$ is a vertex of $P$. If it is not, the original LP must be unbounded.
In general, however, adding constraints to $P$ happens at the expense
of a degraded $\delta$. In particular, the standard reduction of adding a large
box constraint can degrade $\delta$ arbitrarily, hence the constraints must be
added with care. We state the guarantees we can achieve below.

\begin{lemma}[Removing Unboundedness, see Appendix~\ref{sec:removing-unboundedness}]
\label{lem:unb-to-bounded}
Let $P = \set{\vec{x} \in \R^n: A\vec{x} \leq \vec{b}}$ be an $n$-dimensional
pointed polyhedron with $m$ constraints. Let $\vec{a}_1,\dots,\vec{a}_m$ denote the rows
of $A$ and $b_{\rm max} = \max_{i \in [m]} |\vec{b}_i|/\|\vec{a}_i\|$.
\begin{enumerate}
  \item Assume that $P$ satisfies the \emph{local $\delta$-distance property} and
    that $I \subseteq [m]$, $|I|=n$, indexes the rows of a feasible basis.  Letting
    $\vec{w} = -1/n \sum_{i \in I} \vec{a}_i/\|\vec{a}_i\|$, we have that
    \begin{align*}
    P' = \set{\vec{x} \in \R^n: A\vec{x} \leq \vec{b}, \quad
                  \pr{\vec{w}}{\vec{x}} \leq
      n b_{\rm max}/\delta } \text{ ,}
    \end{align*}
    is a polytope that is LP equivalent to $P$ and satisfies the
    \emph{local $\delta^2/(2n)$-distance property}.

  \item Assume that $A$ satisfies the \emph{global $\delta$-distance property}.
    Then
    \begin{align*}
    P' = \{\vec{x} \in \R^n:
    -n \|\vec{a}_i\|b_{\rm max}/\delta - 1
          \leq \pr{\vec{a}_i}{\vec{x}} \leq \vec{b}_i,
    \quad \forall i \in [m]\} \text{ ,}
    \end{align*}
    is a polytope that is LP equivalent to $P$ and satisfies the \emph{global
    $\delta$-distance} property.
\end{enumerate}
\end{lemma}

Finally, we use standard techniques for reducing feasibility to phase $2$ type
optimization. As this generally requires pivoting over infeasible bases, we will 
require global instead of local properties here. Interestingly, for LPs with
bounded subdeterminants, we get that the number of simplex pivots is completely
independent of the number of constraints.

\begin{theorem}[Feasibility via Shadow Simplex, see Appendix~\ref{sec:feasibility}]
Let $P = \set{\vec{x} \in \R^n: A\vec{x} \leq b}$ be an $n$-dimensional
polyhedron whose constraint matrix has full column rank and satisfies the global
$\delta$-distance property.  Then a feasible solution to $P$ can be computed
using an expected $O(mn^3/\delta \ln(n/\delta))$ shadow simplex pivots.
Furthermore, if $A$ is integral and has subdeterminants bounded by $\Delta$, a
feasible solution can be computed using an expected $O(n^5\Delta^2
\ln(n\Delta)))$ shadow simplex pivots.
\end{theorem}

\paragraph*{\emph{Shadow Simplex Method}}
Our main technical contribution is a new analysis and variant of the
\emph{shadow simplex method}, which utilizes (rather unexpectedly) an approach
developed in~\cite{BD14} for navigating over the Voronoi graph of a Euclidean
lattice (see related work section). 

The shadow simplex has been at the heart of many theoretical attempts to explain
the surprising efficiency of the simplex method in practice.  It has been shown
to give polynomial bounds for the simplex method over random and smoothed linear
programs~\cite{MR868467,MR2145860,MR2529774}.  As mentioned above, Brunsch and
R{\"o}glin~\cite{BR13} already showed that it yields short paths for the
polyhedra we consider here.  

At a high level, the shadow simplex over a polyhedron $P$ works as follows.
Given an initial objective function $\vec{c}$,
a vertex $\vec{v}$ of $P$ which maximizes this objective,
and a target objective function $\vec{d}$,
the shadow simplex interpolates between the objective functions $\vec{c}$ and $\vec{d}$
and performs a pivot step whenever the optimal vertex changes
(hence the alternative name \emph{parametric} simplex method referring
to the parameterization $\vec{c}(\lambda) = (1 - \lambda) \vec{c} + \lambda \vec{d}$ of the objective function,
where $\lambda$ grows from $0$ to $1$ over the course of the algorithm).

Traditionally, this method is understood and analyzed with a primal interpretation:
The polyhedron $P$ is orthogonally projected onto the $2$-dimensional plane spanned by $\vec{c}$ and $\vec{d}$
(hence the term ``shadow''),
and the algorithm is understood in terms of the boundary of the projection $P'$.
The optimal vertices for $\vec{c}$ and $\vec{d}$ project to the boundary of $P'$,
and as long as $\vec{c}$ and $\vec{d}$ are in sufficiently general position,
edges of $P'$ lift to edges of $P$
so that the boundary can be followed efficiently by an algorithm that performs simplex pivots in the original space.
The number of pivot steps is then typically bounded in terms of the lengths of edges or in terms of angles between edges of $P'$.

Our analysis is substantially different and based on a \emph{dual} perspective:
The shadow simplex method follows the line segment $[\vec{c}, \vec{d}]$ through the normal fan of $P$,
pivoting whenever the segment crosses into a different $n$-dimensional normal cone.
We express the number of crossings, that is, the number of intersections between $[\vec{c}, \vec{d}]$
and the facets of the normal fan of $P$, in terms of certain surface area
measures of translates of the normal fan. The bounds we obtain on the number of
intersections are stated below.

\begin{theorem}[Intersection bounds, see Lemmas~\ref{lem:exp-partition-hitting} and~\ref{lem:scaling-partition-hitting}]
  \label{thm:intersection-bounds}
  Let $\cT = (C_1,\ldots,C_k)$ be a partition of a cone $\Sigma$ into
  polyhedral $\tau$-wide cones. Let $\vec{c}, \vec{d} \in \R^n$ and let $X
  \in\R^n$ be exponentially distributed on $\Sigma$.
  \begin{enumerate}
    \item
      The expected number of
      facets hit by the shifted line segment $[\vec{c} + X, \vec{d} + X]$ satisfies
      \[
        \E[ |\partial\cT \cap [\vec{c} + X, \vec{d} + X]| ]
        \leq \frac{\|\vec{d} - \vec{c}\|}{\tau} \text{ .}
      \]

    \item
      Let $\alpha \in (0,1)$. Then
      \[
        \E[|\partial \cT \cap [\vec{c} + \alpha X, \vec{c} + X]|] \leq
        \frac{2n}{\tau} \ln{\frac{1}{\alpha}} \text{ .}
      \]
  \end{enumerate}
\end{theorem}

To achieve the above bounds, the main idea is to relate the probability that the
above random line segments pass through a normal cone to the probability that
the associated perturbation vector lands in the cone (or some joint shift).
Under the $\tau$-wideness condition, we can in fact uniformly upper bound these
proportionality factors. Since the jointly shifted normal cones are all
disjoint, we can deduce the desired bounds from the fact that the sum of their
measures is $\leq 1$.
  
We compose these bounds in a way that also departs from the classical template
by using three consecutive shadow simplex paths instead of just one.  For given
vertices $\vec{v}$ and $\vec{w}$ we first pick objectives $\vec{c}$ and
$\vec{d}$ that are ``deep inside'' the respective normal cones. From here, we
sample an exponentially distributed perturbation vector $X$ and traverse three
paths through the normal fan in sequence:
\[
  \vec{c} \quad\stackrel{(a)}{\longrightarrow}\quad \vec{c} + X \quad\stackrel{(b)}{\longrightarrow}\quad \vec{d} + X
          \quad\stackrel{(c)}{\longrightarrow}\quad \vec{d}
\]
The perturbation $X$ will be quite large
and hence almost always large enough to push $\vec{c}$ and $\vec{d}$ away from their normal cones.
Indeed, the high level intuition behind our path is that in order to avoid unusually long paths from
$\vec{c}$ to $\vec{d}$, we first travel to a ``random intermediate location''.

We note that in phases $(a)$ and $(c)$, randomness is only used to perturb one
of the objectives. As far as we are aware, this paper provides the first
successful analysis of the shadow simplex path in this setting.  Furthermore,
this extension is crucial to achieving our improved diameter bound. Previous
algorithms were constrained to random perturbations that kept $\vec{c}$ and
$\vec{d}$ inside their respective normal cones, making the amount of randomness
they could take advantage of much smaller. 

We now use the bounds from Theorem~\ref{thm:intersection-bounds} to derive the
diameter bound.

\begin{theorem}\label{thm:diameter-bound}
Let $P \subseteq \R^n$ be a pointed full dimensional polyhedron with $\tau$-wide
outer normal cones. Then $P$ has diameter bounded by
$\frac{8n}{\tau}(1+\ln 1/\tau)$.
\end{theorem}
\begin{proof}
Let $\vec{v}_1,\vec{v}_2$ be vertices of $P$ with outer normal cones
$N_{\vec{v}_1},N_{\vec{v}_2}$. Let $\vec{c}_1,\vec{c}_2 \in \bS^{n-1}$ satisfy
$\vec{c}_i+\tau \cB_2^n \subseteq N_{\vec{v}_i}$, $i \in \set{1,2}$. Let
$\Sigma = N(P)$ denote the support of the normal fan of $P$, and let $X$ be
exponentially distributed over $\Sigma$.

We will construct a path from $\vec{v}_1$ to $\vec{v}_2$ by following
the sequence of vertices optimizing the objectives in the segments
$[s \vec{c}_1, s \vec{c}_1 + X]$, $[s \vec{c}_1 + X, s \vec{c}_2 + X]$,
$[s \vec{c}_2 + X, s \vec{c}_2]$, where $s > 0$ is a scalar
to be chosen later. We will condition on the event that $\|X\| \leq 2n$. Since
$\E[\|X\|] = n$ (see Lemma~\ref{lem:exp-moments}), by Markov's inequality this
occurs with probablity at least $1/2$. Under this event, by $\tau$-wideness, we
will not pivot in the segments $[s \vec{c}_1, s \vec{c}_1 +
\frac{s\tau}{2n} X]$ and $[s \vec{c}_2 + \frac{s\tau}{2n} X, s \vec{c}_2]$.
Using Theorem~\ref{thm:intersection-bounds}, the number of pivots along the segments $[s
\vec{c}_1 + \frac{s\tau}{2n} X, s\vec{c}_1 + X]$, $[s \vec{c}_1 + X, s \vec{c}_2
+ X]$, $[s \vec{c}_2 + X, s \vec{c}_2 +  \frac{s\tau}{2n} X]$, is bounded by
\[
  \frac{\left(\frac{s\|\vec{c}_2-\vec{c}_1\|}{\tau} +
    \frac{4n}{\tau} \ln\left(\frac{2n}{s\tau}\right) \right)}{\Pr[\|X\| \leq 2n]}
  \leq
    2 \left(\frac{s\|\vec{c}_2-\vec{c}_1\|}{\tau} + \frac{4n}{\tau}
    \ln\left(\frac{2n}{s\tau}\right) \right) \text{.}
\]
Setting $s = \frac{4n}{\|\vec{c}_2-\vec{c}_1\|}$, the above bound becomes
\[
\frac{8n}{\tau}\left(1+\ln\left(\frac{\|\vec{c}_2-\vec{c}_1\|}{2\tau}\right)\right) 
\leq \frac{8n}{\tau}\left(1+ \ln \frac{1}{\tau} \right)
\text{,\quad as needed.} \qedhere
\]
\end{proof}

\paragraph*{Related Work}

In a surprising connection, we borrow techniques developed in a recent work of
Bonifas and the first named author~\cite{BD14} for a totally different purpose,
namely, for solving the Closest Vector Problem with Preprocessing on Euclidean
lattices. In~\cite{BD14}, a $3$-step ``perturbed'' line path was analyzed to
navigate over the Voronoi graph of the lattice, where lattice points are
connected if their associated Voronoi cells touch in a facet. 

In the current work, we show a strikingly close analogy between analyzing the
number of intersections of a random straight line path with a Voronoi tiling of
space and the intersections of a shadow simplex path with the normal fan of a
polyhedron.  This unexpected connection makes us hopeful that these ideas may
have even broader applicability.   

\paragraph*{Organization}
In section~\ref{sec:prelims}, we regroup all the necessary notation and
definitions. In section~\ref{sec:optimization}, we give our shadow simplex based
optimization algorithm. In section~\ref{sec:int-bounds}, we prove the
intersection bounds for our shadow simplex method. In
section~\ref{sec:shadow-simplex-alg}, we show how to implement shadow simplex
pivots and how to deal with degeneracy. In
section~\ref{sec:removing-unboundedness}, we give our reductions from unbounded
$\delta$-wide LPs to bounded ones. In section~\ref{sec:feasibility}, we give our
shadow simplex based algorithms for LP feasibility. In
section~\ref{sec:perfect-matching}, we give lower bounds on the width of the
normal fan of the perfect matching polytope. Missing proofs can be found in
Appendix~\ref{sec:appendix-proofs}.

\section{Notation and definitions}
\label{sec:prelims}

For vectors $\vec{x},\vec{y} \in \R^n$, we let $\pr{\vec{x}}{\vec{y}} =
\sum_{i=1}^n \vec{x}_i \vec{y}_i$ denote their inner product.  We let
$\|\vec{x}\| = \sqrt{\pr{\vec{x}}{\vec{x}}}$ denote the Euclidean norm, $\cB_2^n =
\set{\vec{x} \in \R^n: \|\vec{x}\| \leq 1}$ the unit ball and $\bS^{n-1} =
\partial \cB_2^n$ the unit sphere. We denote the linear span of a set $A \subseteq
\R^n$ by $\linsp(A)$.  We use the notation ${\rm I}[\vec{x} \in A]$ for the
indicator of $A$, that is ${\rm I}[\vec{x} \in A]$ is $1$ if $\vec{x} \in A$ and
$0$ otherwise. For a set of scalars $S \subseteq \R$, we write $SA =
\set{s\vec{a}: s \in S, \vec{a} \in A}$. For two sets $A,B \subseteq \R^n$, we
define their Minkowski sum $A+B = \set{\vec{a}+\vec{b}: \vec{a} \in A, \vec{b}
\in B}$. We let $d(A,B) = \inf \set{\|\vec{x}-\vec{y}\|: \vec{x} \in A, \vec{y}
\in B}$, denote the Euclidean distance between $A$ and $B$. For vectors
$\vec{a},\vec{b} \in \R^n$ we write $[\vec{a},\vec{b}]$ for the closed line
segment and $[\vec{a},\vec{b})$ for the half-open line segment from $\vec{a}$ to $\vec{b}$.

\begin{definition}[Cone]
A cone $\Sigma \subseteq \R^n$ satisfies the following three properties:
\begin{itemize}
\item $\vec{0} \in \Sigma$. 
\item $\vec{x}+\vec{y} \in \Sigma$ if $\vec{x}$ and $\vec{y}$
are in $\Sigma$. 
\item $\lambda \vec{x} \in \Sigma$ if $\vec{x} \in \Sigma$ and
$\lambda \geq 0$. 
\end{itemize}
\end{definition}

For vectors $\vec{y}_1,\dots,\vec{y}_k \in \R^n$, we define
the closed cone they generate as
\[
\cone(\vec{y}_1,\dots,\vec{y}_k) = \set{\sum_{i=1}^m \lambda_i \vec{y}_i:
\lambda_i \geq 0, i \in [m]} \text{.}
\]
A cone is polyhedral if it can be generated by a finite number of vectors, and
is simplicial if the generators are linearly independent. By convention, we let
$\cone(\emptyset) = \vec{0}$.
A simplicial cone has the $\delta$-distance property if its extreme rays
satisfy the $\delta$-distance property.\footnote{The $\delta$-distance property is invariant under scaling,
so the choice of generators of the extreme rays is irrelevant.}

For a convex set $K \subseteq \R^n$, a subset $F \subseteq K$ is a face of $K$,
if for all $\vec{x},\vec{y} \in K$, $\lambda \vec{x} + (1-\lambda) \vec{y} \in
F$, $\lambda \in [0,1]$, implies that $\vec{x},\vec{y} \in F$. For a simplicial
cone $C$, we note that its faces are exactly all the subcones generated by any
subset of the generators of $C$.


A set of cones $\cT = \set{C_1,\dots,C_k}$ is an $n$-dimensional \emph{cone
partition} if:
\begin{itemize}
\item Each $C_i \subseteq \R^n$, $i \in [k]$, is a closed $n$-dimensional
cone.
\item Any two cones $C_i,C_j$, $i \neq j$, meet in a shared face.  
\item The \emph{support} of $\cT$, ${\rm sup}(\cT) \eqdef \cup_{i \in [k]} C_i$, 
is a closed cone.
\end{itemize}
We say that $F$ is a face of $\cT$ if it is a face of one of its contained cones.
A cone partition $\cT$ is $\tau$-wide if every $C_i$ is $\tau$-wide.
It is simplicial if every $C_i$ is simplicial.
In this case, we also call $\cT$ a \emph{cone triangulation}.
A cone triangulation satisfies the local $\delta$-distance property if every $C_i$ satisfies it.
We define the boundary of $\cT$, $\partial \cT = \cup_{i=1}^k \partial C_i$.
We say that a cone triangulation $\cT$ \emph{triangulates} a cone partition $\cP$
if $\cT$ and $\cP$ have the same support
and every cone $C \in \cT$ is generated by a subset of the extreme rays of some cone of $\cP$.
This means that $\cT$ partitions (``refines'') every cone of $\cP$ into simplicial cones.

\subsection{Exponential distribution}
We say that a random variable $X \in \R^n$ is exponentially distributed
on a cone $\Sigma$ if
\[
\Pr[X \in S] = \int_S \zeta_\Sigma(\vec{x}) \d\vec{x} 
\text{,}
\]
where $\zeta_{\Sigma}(\vec{x}) = c_\Sigma e^{-\|\vec{x}\|} {\rm I}[\vec{x} \in
\Sigma]$.
A standard computation, which we include for completeness,
yields the normalizing constant and the expected norm.

\begin{lemma}\label{lem:exp-moments}
The normalizing constant $c_\Sigma^{-1} = n! \vol_n(\cB_2^n \cap
\Sigma)$. For $X$ exponentially distributed on $\Sigma$, we have that $\E[\|X\|]
= n$.
\end{lemma}
\begin{proof}
  See Appendix~\ref{sec:appendix-proofs}.
\end{proof}

\section{Optimization}
\label{sec:optimization}

While bounding the number of intersections of line segments $[\vec{c}, \vec{d}]$
with the facets of the normal fan of $P = \set{ \vec{x} \in \R^n : A\vec{x} \leq \vec{b} }$
is sufficient to obtain existential bounds on the diameter of $P$,
we also need to be able to efficiently compute the corresponding pivots to obtain efficient algorithms.
The following summarizes the required results,
the technical details of which are found in Appendix~\ref{sec:shadow-simplex-alg}.

\begin{theorem}[Shadow simplex, see Theorem~\ref{thm:shadow-simplex}]
  Let $P = \set{ \vec{x} \in \R^n : A\vec{x} \leq \vec{b} }$ be pointed,
  $\vec{c}, \vec{d} \in \R^n$, and $B$ an optimal basis for $\vec{c}$.
  If every intersection of $[\vec{c}, \vec{d})$ with a facet $F$ of a cone spanned by a feasible basis of $P$
  lies in the relative interior of $F$,
  the Shadow Simplex can be used to compute an optimal basis for $\vec{d}$
  in $O(mn^2 + Nmn)$ arithmetic operations,
  where $N$ is the number of intersections of $[\vec{c},\vec{d}]$
  with some triangulation $\cT$ of the normal fan of $P$,
  where $\cT$ contains the cone spanned by the initial basis $B$.
\end{theorem}

As explained in Section~\ref{sec:results},
we want to follow segments $[\vec{c}, \vec{c} + X]$,
$[\vec{c} + X, \vec{d} + X]$, $[\vec{d} + X, \vec{d}]$ in the normal fan.
Our intersection bounds from Theorem~\ref{thm:intersection-bounds}
are not quite sufficient to bound the number of steps on the first and last segments entirely.
This is easily dealt with for the first segment,
because we can control the initial objective function $\vec{c}$
so that it lies deep in the initial normal cone.

For the final segment, we follow the approach of Eisenbrand and Vempala~\cite{EV14},
who showed that if $A$ satisfies the \emph{global} $\delta$-distance property,
then an optimal facet for $\vec{d}$ can be derived
from a basis that is optimal for some $\vec{\tilde d}$ with $\|\vec{d} - \vec{\tilde d}\| \leq \frac{\delta}{n}$.
Recursion can then be used on a problem of reduced dimension to move from $\vec{\tilde d}$ to $\vec{d}$.
We strengthen their result (thereby answering a question left open by~\cite{EV14})
and show that the \emph{local} $\delta$-distance property is sufficient to get the same result
as long as $\|\vec{d} - \vec{\tilde d}\| \leq \frac{\delta}{n^2}$.%
\footnote{In the final bound, the loss of a factor $n$ here disappears inside a logarithm.}

\begin{definition}
Let $F$ be a face of a cone triangulation $\cT$ and let $\vec{x}$ be a vector in
the support of $\cT$.  Let $G = \cone(\vec{x}_1, \ldots, \vec{x}_k)$, $\|\vec{x}_i\| = 1$,
be the minimal face of $\cT$ that contains $\vec{x}$ and consider the unique conic
combination $\vec{x} = \lambda_1 \vec{x}_1 + \dots + \lambda_k \vec{x}_k$. We 
define
\[
\alpha_F(\vec{x}) := \sum_{i : \vec{x}_i \not\in F} \lambda_i
\]
In particular, $\alpha_F(\vec{x}) \geq 1$ if $\vec{x}$ is a unit vector and the
minimal face containing it is disjoint from $F$, and $\alpha_F(\vec{x}) = 0$ if
$\vec{x} \in F$.
\end{definition}

\begin{lemma}\label{lem:delta-dist-equiv}
Let $\vec{x}_1,\dots,\vec{x}_m \in \bS^{n-1}$ be a set of vectors.  Then the
following are equivalent:
\begin{enumerate}
\item $\vec{x}_1,\dots,\vec{x}_m$ satisfy the $\delta$-distance property.
\item $\forall~ I \subseteq [m]$ for which $\set{\vec{x}_i: i \in
I}$ are linearly independent and $\forall~ (a_i \in \R: i \in I)$
\[
\|\sum_{i \in I} a_i \vec{x}_i\| \geq \delta \max_{i \in I} |a_i| \text{ .}
\]
\end{enumerate}
\end{lemma}
\begin{proof}
  See Appendix~\ref{sec:appendix-proofs}.
\end{proof}

\begin{lemma} \label{lemma:distance-from-cone-faces}
  Let $F$ be a cone of an $n$-dimensional cone triangulation $\cT$ satisfying the
  local $\delta$-distance property. Let $\vec{x}$ be a point in the support of $\cT$.
  Then $d(\vec{x},F) \geq \alpha_F(\vec{x}) \cdot \frac{\delta}{n}$.
\end{lemma}
\begin{proof}
Let $\vec{y} \in F$ be the (unique) point with $d(\vec{x},\vec{y}) =
d(\vec{x},F)$.  Note that by convexity, the segment $[\vec{x},\vec{y}]$ is
contained in the support of $\cT$. By considering the cones of $\cT$ that
contain points on the segment $[\vec{x},\vec{y}]$, we obtain a sequence of
points
\[
        \vec{x} = \vec{x}_0, \vec{x}_1, \ldots, \vec{x}_r = \vec{y}
\]
on the segment $[\vec{x},\vec{y}]$ and (full-dimensional) cones $G_1, \ldots,
G_r$ such that
\[
    G_i \cap [\vec{x},\vec{y}] = [\vec{x}_{i-1}, \vec{x}_i].
\]
Since $\alpha_F(\vec{y}) = 0$, the result of the lemma follows immediately from
the claim that
\[
d(\vec{x}_{i-1}, \vec{x}_i) \geq |\alpha_F(\vec{x}_{i-1}) - \alpha_F(\vec{x}_i)|
\cdot \frac{\delta}{n},
\]
which we will now prove.

Fix some $G_i = \cone(\vec{y}_1, \ldots, \vec{y}_n)$. By relabelling, we may
assume that $\cone(\vec{y}_1,\dots,\vec{y}_k) = G_i \cap F$ (since $G_i$ and $F$
are both faces of $\cT$), for some $0 \leq k \leq n$ (if $k=0$ then $G_i \cap F
= \vec{0}$).

For every $\vec{z} \in G_i$, the minimal cone containing $\vec{z}$ is a face of
$G_i$. Therefore, using the unique conic combination $\vec{z} = \sum_{i=1}
\lambda_i \vec{y}_i$, we have that $\alpha_F(\vec{z}) = \sum_{k < i \leq n}
\lambda_i$.

Writing $\vec{x}_{i-1} = \sum_{i=1}^n a_i \vec{y}_i$ and $\vec{x}_i =
\sum_{i=1}^n b_i \vec{y}_i$, by Lemma~\ref{lem:delta-dist-equiv} we have that
\begin{align*}
d(\vec{x}_{i-1},\vec{x}_i)
&\geq \delta \max_{1 \leq i \leq n} |a_i-b_i|
\geq \delta \max_{k < i \leq n} |a_i-b_i|
\geq \frac{\delta}{n} \sum_{k < i \leq n} |a_i - b_i| \\
&\geq \frac{\delta}{n} |\alpha_F(\vec{x}_{i-1})-\alpha_F(\vec{x}_i)| \text{,}
\end{align*}
which completes the proof of the claim.
\end{proof}

\begin{lemma}
  \label{lem:snap-to-cone}
  Let $F$ be a cone of a triangulation $\cT$ satisfying the local $\delta$-distance property and
  let $\vec{x}$ be a point in the support of $\cT$ with $d(\vec{x},F) \leq \frac{\delta}{n^2}$.
  Let $G = \cone(\vec{x}_1, \ldots, \vec{x}_n)$, $\|\vec{x}_i\| = 1$,
  be a cone of $\cT$ containing $\vec{x}$ and let
  \[
    \vec{x} = \lambda_1 \vec{x}_1 + \dots + \lambda_n \vec{x}_n
  \]
  be the corresponding conic combination. Then for every $i \in [n]$ with
  $\lambda_i > \frac{1}{n}$ one has $\vec{x}_i \in F$.
\end{lemma}
\begin{proof}
Suppose there is some $i$ with $\lambda_i > \frac{1}{n}$ and $\vec{x}_i
\not\in F$.  Then $\alpha_F(\vec{x}) > \frac{1}{n}$ and by
Lemma~\ref{lemma:distance-from-cone-faces} we get $d(\vec{x},F) >
\frac{\delta}{n^2}$, which is a contradiction.
\end{proof}

For the recursion on a facet,
we let $\pi_i(\vec{x}) := \vec{x} - \frac{\pr{\vec{x}}{\vec{a}_i}}{\pr{\vec{a}_i}{\vec{a}_i}} \vec{a}_i$
be the orthogonal projection onto the subspace orthogonal to $\vec{a}_i$
and we let $F_i$ be the facet of $P$ defined by $\pr{\vec{a}_i}{\vec{x}} = \vec{b}_i$.
\begin{lemma}
  \label{lem:delta-distance-projection}
  Let $\vec{v}_1, \ldots, \vec{v}_k \in \R^n$ be linearly independent vectors
  that satisfy the $\delta$-distance property and let $\pi$ be the orthogonal projection
  onto the subspace orthogonal to $\vec{v}_k$.
  Then $\pi(\vec{v}_1), \ldots, \pi(\vec{v}_{k-1})$ satisfy the $\delta$-distance property.
\end{lemma}
\begin{proof}
  See Appendix~\ref{sec:appendix-proofs}.
\end{proof}
This Lemma, which was already used by~\cite{EV14},
implies that if $P$ satisfies the local $\delta$-distance property
then so does $F_i$, where the definition of local $\delta$-distance is understood
relative to the affine hull of $F_i$,\footnote{Alternatively,
one can apply a rotation and translation so that $F_i$ lies in
the subspace $\R^{n-1}$ spanned by the first $n-1$ coordinates.
The rotation does not affect the $\delta$-distance property,
and we can then treat $F_i$ as a polytope in $\R^{n-1}$.}
because the normal vectors of $F_i$ arise from orthogonal projections of the normal vectors of $P$.

\begin{algorithm}[]
  \DontPrintSemicolon
  \KwIn{polytope $P = \set{ \vec{x} \in \R^n : A\vec{x} \leq \vec{b} }$, $\delta > 0$,
    feasible basis $B$,
    $\vec{d} \in \R^n$}
  \KwOut{optimal basis $B \subset [m]$ for $\vec{d}$}
  $\vec{c} \gets \sum_{i \in B} \frac{\vec{a}_i}{\|\vec{a}_i\|}$, $\vec{d} \gets 2 \frac{\vec{d}}{\|\vec{d}\|}$\;
  Sample $X \in \R^n$ from the exponential distribution conditioned on $\|X\| \leq 2n$\;
  Follow segments $[\vec{c}, \vec{c} + X]$, $[\vec{c} + X, \vec{d} + X]$,
    $[\vec{d} + X, \vec{d} + \frac{\delta}{2n^3} X]$
    using Shadow Simplex\;
  Find $\lambda_i$ such that
    $\vec{\tilde d} := \vec{d} + \frac{\delta}{2n^3} X = \sum_{i \in B} \lambda_i \frac{\vec{a}_i}{\|\vec{a}_i\|}$
    where $B$ is the current basis\;
  Choose $i^\star$ such that $\lambda_{i^\star} > \frac{1}{n}$\;
  $B' \gets $ optimal basis of $F_{i^\star}$ for $\pi_{i^\star}(\vec{d})$,
    obtained by recursion starting at $B\setminus\set{i^\star}$\;
  \Return $B' \cup \set{i^\star}$\;
  \caption{Optimization\label{alg:optimization}}
\end{algorithm}

\begin{theorem}
  \label{thm:optimization-detail}
  If $P$ satisfies the local $\delta$-distance property,
  then Algorithm~\ref{alg:optimization} correctly computes an optimal basis for $\vec{d}$
  using an expected $O(n^3/\delta \ln(n/\delta))$ shadow simplex pivots.
\end{theorem}
\begin{proof}
  For correctness,
  let $\cT$ be some triangulation of the normal fan of $P$
  and let $C$ be a cone in $\cT$ that contains $\vec{d}$.
  We have $\|\frac{\delta}{2n^3} X\| \leq \frac{\delta}{n^2}$
  and therefore $d(\vec{\tilde d}, C) \leq d(\vec{\tilde d}, \vec{d}) \leq \frac{\delta}{n^2}$.
  Furthermore, $\|\vec{\tilde d}\| \geq \|\vec{d}\| - \frac{\delta}{n^2} > 1$
  implies that $\sum_{i \in B} \lambda_i > 1$
  so that there is some $i$ with $\lambda_i > \frac{1}{n}$.
  Applying Lemma~\ref{lem:snap-to-cone} yields that $\vec{a}_{i^\star}$
  is a generator of $C$,
  which means that $i^\star$ is contained in some optimal basis for $\vec{d}$.
  This implies that recursion on $F_{i^\star}$ yields the correct result.

  In order to bound the number of pivots,
  let $C$ be the cone of the initial basis
  and observe that $\vec{c} + \delta \cB_2^n \subseteq C$
  by the proof of Lemma~\ref{lem:dist-to-width}.
  Hence the segment $[\vec{c} + \frac{\delta}{2n} X)$ does not cross a facet of
  the triangulation $\cT_1$ of the normal fan that is implicitly used by the first
  leg of the shadow simplex path.

  If $X$ were exponentially distributed (without the conditioning on $\|X\| \leq 2n$),
  Theorem~\ref{thm:intersection-bounds}
  would bound the expected number of pivot steps along the three segments by
  \[
    \E[ N ] \leq \frac{2n^2}{\delta} \ln \frac{2n}{\delta} + \frac{n \|\vec{d} - \vec{c}\|}{\delta}
      + \frac{2n^2}{\delta} \ln \frac{2n^3}{\delta}
    \leq O(\frac{n^2}{\delta} \ln(\frac{n}{\delta}))
  \]
  Since $\E[\|X\|] = n$ we have $\Pr[ \|X\| \leq 2n ] \geq \frac{1}{2}$ by Markov's inequality
  and therefore
  \[
    \E[ N ~|~ \|X\| \leq 2n ] \leq 2 \E[ N ] \leq O(\frac{n^2}{\delta} \ln(\frac{n}{\delta})).
  \]
  The bound on the total expected number of pivot steps follows from the depth $n$ of recursion.
\end{proof}

\section{Intersection Bounds and Diameter Bounds}
\label{sec:int-bounds}

\begin{lemma}
  \label{lem:exp-cone-hitting}
  Let $C$ be a polyhedral cone containing $\vec{u} + \tau \cB_2^n$, where
  $\|\vec{u}\|=1$. Let $\vec{c}, \vec{d} \in \R^n$ and let $X \in \R^n$ be
  exponentially distributed on a full dimensional cone $\Sigma \ni \vec{u}$. Then
  the expected number of times the shifted line segment $[\vec{c} + X, \vec{d} +
  X]$ hits the boundary of $C$ is at most
  \[
      \E[|\partial C \cap [\vec{c} + X, \vec{d} + X]|]
        \leq \frac{\|\vec{d} - \vec{c}\|}{\tau}
          \int_0^1 \int_{(C -((1-\lambda)c + \lambda d)) \cap \Sigma} \zeta_\Sigma(\vec{x}) \d\vec{x} \d\lambda
    \]
\end{lemma}
\begin{proof}
Let $F$ be a facet of $C$. Note that with probability $1$, the line segment
$[\vec{c}+X,\vec{d}+X]$ passes through $F$ at most once. By linearity, we see
that
\begin{equation}
\label{eq:eci-1}
\E[|\partial C \cap [\vec{c} + X, \vec{d} + X]|] = \sum_{F \text{ facet of } C}
\Pr[(F \cap [\vec{c}+X,\vec{d}+X]) \neq \emptyset] \text{.}
\end{equation}
We now bound the crossing probability for any facet $F$.

We first calculate the hitting probability as
\begin{align}
\label{eq:eci-2}
\begin{split}
	\Pr[ F \cap [\vec{c} + X, \vec{d} + X] \neq \emptyset]
		&= \Pr[ X \in -[\vec{c},\vec{d}] + F ] \\
		&= \int_{-[\vec{c},\vec{d}] + F} \zeta_\Sigma(\vec{x}) \d\vec{x} \\
		&= |\pr{\vec{n}}{\vec{d}-\vec{c}}|
			 \int_0^1 \int_{F - ((1-\lambda)\vec{c} + \lambda \vec{d})}
								\zeta_\Sigma(\vec{x}) \d\vol_{n-1}(\vec{x}) \d\lambda \\
    &\leq \|\vec{d}-\vec{c}\| \int_0^1 \int_{(F - ((1-\lambda)\vec{c} + \lambda
\vec{d})) \cap \Sigma} c_\Sigma e^{-\|\vec{x}\|} \d\vol_{n-1}(\vec{x}) \d\lambda
\end{split}
\end{align}
where $\vec{n} \in \R^n$ is a unit normal vector to $F$. Bounding the hitting
probability therefore boils down to bounding the measure of a shift of the facet
$F$. Letting $h = |\pr{\vec{n}}{\vec{u}}| \geq \tau$ (which holds by assumption on
$\vec{u}$), for any shift $\vec{t} \in \R^n$ we have that
\begin{align}
\label{eq:eci-3}
  \begin{split}
   \int_{(F + \vec{t} + \cone(\vec{u})) \cap \Sigma}
       e^{-\|\vec{x}\|} \d\vec{x}
    &\geq \int_{((F + \vec{t}) \cap \Sigma) + \cone(\vec{u})} e^{-\|\vec{x}\|} \d\vec{x}
              \quad \left(\text{ since $\vec{u} \in \Sigma$ }\right) \\
    &= \int_0^\infty \int_{((F + \vec{t}) \cap \Sigma) + \frac{r}{h} \vec{u}}
              e^{-\|\vec{x}\|} \d\vol_{n-1}(\vec{x}) \d r \\
    &= \int_0^\infty \int_{(F + \vec{t}) \cap \Sigma}
              e^{-\|\vec{x} + \frac{r}{h} \vec{u}\|} \d\vol_{n-1}(\vec{x}) \d r \\
    &\geq \int_0^\infty e^{-r/h} \d r \int_{(F + \vec{t}) \cap \Sigma} e^{-\|\vec{x}\|} \d\vol_{n-1}(\vec{x}) \\
    &\geq \tau \int_{(F + \vec{t}) \cap \Sigma} e^{-\|\vec{x}\|} \d\vol_{n-1}(\vec{x})
\end{split}
\end{align}
The lemma now follows by combining
\eqref{eq:eci-1},\eqref{eq:eci-2},\eqref{eq:eci-3}, using the fact that
the $F + \cone(\vec{u})$ partition the cone $C$ up to sets of measure $0$.
\end{proof}

\begin{lemma}\label{lem:exp-partition-hitting}
  Let $\cT = (C_1,\ldots,C_k)$ be a partition of a cone $\Sigma$ into
  polyhedral $\tau$-wide cones. Let $\vec{c}, \vec{d} \in \R^n$ and let $X
  \in\R^n$ be exponentially distributed on $\Sigma$. Then the expected number of
  facets hit by the shifted line segment $[\vec{c} + X, \vec{d} + X]$ satisfies
    \[
      \E[ |\partial\cT \cap [\vec{c} + X, \vec{d} + X]| ]
        \leq \frac{\|\vec{d} - \vec{c}\|}{\tau} \text{ .}
    \]
\end{lemma}
\begin{proof}
  Using Lemma~\ref{lem:exp-cone-hitting}, we bound
  \begin{align*}
    \E[ |\partial\cT \cap [\vec{c} + X, \vec{d} + X]| ]
      &\leq \sum_{i=1}^k \E[|\partial C_i \cap [\vec{c} + X, \vec{d} + X]|] \\
      &\leq \sum_{i=1}^k \frac{\|\vec{d} - \vec{c}\|}{\tau} \int_0^1
        \int_{(C_i-((1-\lambda)\vec{c}+\lambda \vec{d})) \cap \Sigma} \zeta_{\Sigma}(\vec{x})\d\vec{x} \d\lambda \\
      &\leq \frac{\|\vec{d} - \vec{c}\|}{\tau} \int_0^1
        \int_{\Sigma} \zeta_{\Sigma}(\vec{x})\d\vec{x} \d\lambda \\
      &\leq \frac{\|\vec{d} - \vec{c}\|}{\tau} \text{,}
  \end{align*}
as needed.

For the furthermore, note that each intersection is overcounted twice in the
summation above, since each facet belongs to exactly two cones in the partition.
\end{proof}

We will need the following simple lemma about the exponential distribution.

\begin{lemma}\label{lem:exp-variance} 
Let $Y$ be exponentially distributed on $\R_+$.
Then for any $c \in \R$, $\E[|Y-c|] \geq |c|/2$.
\end{lemma}
\begin{proof}
Since $Y \geq 0$, the inequality is trivial if $c \leq 0$. Hence we
may assume that $c \geq 0$. Using integration by parts, we have that
\begin{align*}
\E[|Y-c|] &= \int_0^c (c-x)e^{-x}\d x + 
                 \int_c^\infty (x-c)e^{-x} \d x \\ 
          &= (x-c)e^{-x}\big|_0^c - \int_0^c e^{-x} \d x
            + (c-x)e^{-x}\big|_c^\infty + \int_c^\infty e^{-x} \d x \\
          &= (x-c+1)e^{-x}\big|_0^c + (c-x-1)e^{-x}\big|_c^\infty
          = 2e^{-c} + c - 1 \text{.}
\end{align*}
We wish to show that $2e^{-c} + c - 1 \geq c/2$, hence it suffices to show
$2e^{-c} + c/2 - 1 \geq 0$ for all $c \geq 0$. This function is minimized at $c
= \ln 4$ where it achieves value $(\ln 4 - 1)/2 > 0$.
\end{proof}

While we could choose $\vec{c}$ and $\vec{d}$ such that $\vec{c} + X$ and
$\vec{d} + X$ lie in the same cone with high probability, this would require us
to choose $\|\vec{d} - \vec{c}\|$ quite large. Instead, we will bound the
number of facets that are hit by the segment $[\vec{c}, \vec{c} + X]$.

\begin{lemma}
  \label{lem:scaling-cone-hitting}
  Let $C \subseteq \R^n$ be a polyhedral cone containing $\vec{u}+\tau \cB_2^n$,
  where $\|\vec{u}\|=1$. Let $\vec{c} \in \R^n$ and $X \in \R^n$ be exponentially
  distributed on a cone $\Sigma \ni \vec{u}$. Then for every $\alpha \in (0,1)$ we
  have
  \[
    \E[|\partial C \cap [\vec{c} + \alpha X, \vec{c} + X]|]
    \leq \frac{2}{\tau} \int_1^{1/\alpha} \frac{1}{s}
      \int_{(C - s\vec{c})\cap\Sigma} \|\vec{x}\| \zeta_{\Sigma}(\vec{x}) \d\vec{x} \d s
  \]
\end{lemma}
\begin{proof}
As in the proof of Lemma~\ref{lem:exp-cone-hitting}, we will
decompose the expectation over the facets of $C$, where we have
\begin{equation}
  \label{eq:sfh-1}
  \E[|\partial C \cap [\vec{c} + \alpha X, \vec{c} + X]|] =
  \sum_{F \text{ facet of } C} \Pr[F \cap [\vec{c} + \alpha X, \vec{c} + X] \neq
  \emptyset]
\end{equation}
Take a facet $F$ of $C$ and let $\vec{n}$ denote a unit normal to $F$
pointing in the direction of the cone (i.e., $\pr{n}{u} > 0$).
\begin{align}
\label{eq:sfh-2}
\begin{split}
\Pr[F \cap [\vec{c} + \alpha X, \vec{c} + X] \neq \emptyset]
  &= \Pr[X \in [1, \frac{1}{\alpha}](F - \vec{c})] \\
  &= \int_1^{1/\alpha} \int_{(F - s\vec{c}) \cap \Sigma}
      |\pr{\vec{n}}{\vec{c}}| \zeta_{\Sigma}(\vec{x}) \d\vol_{n-1}(\vec{x}) \d s \\
  &= \int_1^{1/\alpha} \frac{1}{s} \int_{(F - s\vec{c}) \cap \Sigma}
      |\pr{\vec{n}}{s\vec{c}}| c_\Sigma e^{-\|\vec{x}\|} \d\vol_{n-1}(\vec{x}) \d s
\text{.}
\end{split}
\end{align}
Again, we have to bound an integral over a shifted facet, similar to the proof
of Lemma~\ref{lem:exp-cone-hitting}. Letting $h =
|\pr{\vec{n}}{\vec{u}}| \geq \tau$, we have that
\begin{align}
\label{eq:sfh-3}
\begin{split}
  \int_{(F+\vec{t}+\cone(\vec{u})) \cap \Sigma} \|\vec{x}\| e^{-\|\vec{x}\|} \d\vec{x}
    &\geq \int_{((F+\vec{t}) \cap \Sigma)+\cone(\vec{u})}
      \|\vec{x}\| e^{-\|\vec{x}\|} \d\vec{x}
      \quad \left(\text{since $\vec{u} \in \Sigma$}\right)\\
    &= \int_0^\infty \int_{((F+\vec{t}) \cap \Sigma) + \frac{r}{h} \vec{u}}
      \|\vec{x}\| e^{-\|\vec{x}\|} \d\vol_{n-1}(\vec{x}) \d r \\
    &= \int_0^\infty \int_{(F+\vec{t}) \cap \Sigma}
      \|\vec{x} + \frac{r}{h}\vec{u}\| e^{-\|\vec{x} + \frac{r}{h} \vec{u}\|}
      \d\vol_{n-1}(\vec{x}) \d r \\
    &\geq \int_0^\infty \int_{(F + \vec{t}) \cap \Sigma}
      |\pr{\vec{n}}{\vec{x} + \frac{r}{h}\vec{u}}| e^{-r/h} e^{-\|\vec{x}\|}
      \d\vol_{n-1}(\vec{x}) \d r \\
    &= h^2 \int_0^\infty |\pr{\vec{n}}{\vec{t}}/h + s| e^{-s} \d s
      \int_{(F+\vec{t}) \cap \Sigma} e^{-\|\vec{x}\|} \d\vol_{n-1}(\vec{x}) \\
&\geq \frac{h}{2} \int_{(F+\vec{t}) \cap \Sigma} |\pr{\vec{n}}{\vec{t}}|e^{-\|\vec{x}\|}
 \d\vol_{n-1}(\vec{x}) \quad \left(\text{by Lemma~\ref{lem:exp-variance}}\right) \\
&\geq \frac{\tau}{2} \int_{(F+\vec{t}) \cap \Sigma}|\pr{\vec{n}}{\vec{t}}| e^{-\|\vec{x}\|}
 \d\vol_{n-1}(\vec{x})\\
\end{split}
\end{align}
The Lemma now follows by combining~\eqref{eq:sfh-1},\eqref{eq:sfh-2},\eqref{eq:sfh-3}.
\end{proof}

\begin{lemma}\label{lem:scaling-partition-hitting}
  Let $\cT = (C_1,\ldots,C_k)$ be partition of a cone $\Sigma$ into
  polyhedral $\tau$-wide cones. Let $\vec{c} \in \R^n$ and $\alpha \in (0,1)$ be
  fixed and let $X \in \R^n$ be exponentially distributed over $\Sigma$.  Then
  \[
  \E[|\partial \cT \cap [\vec{c} + \alpha X, \vec{c} + X]|] \leq
  \frac{2n}{\tau} \ln{\frac{1}{\alpha}} \text{ .}
  \]
\end{lemma}
\begin{proof}
By Lemmas \ref{lem:exp-moments} and \ref{lem:scaling-cone-hitting}, 
we have that
\begin{align*}
\E[|\partial \cT \cap [\vec{c} + \alpha X, \vec{c} + X]|]
 &\leq \sum_{i=1}^k \E[|\partial C_i \cap [\vec{c} + \alpha X, \vec{c} + X]|] \\
 &\leq \frac{2}{\tau} \sum_{i=1}^k 
 \int_1^{1/\alpha} \frac{1}{s} \int_{(C_i - s\vec{c})\cap\Sigma}
\|\vec{x}\| \zeta_{\Sigma}(\vec{x}) \d\vec{x} \d s \\
 &\leq \frac{2}{\tau} \int_1^{1/\alpha} \frac{1}{s} \int_{\Sigma}
\|\vec{x}\| \zeta_{\Sigma}(\vec{x}) \d\vec{x} \d s \\
 &\leq \frac{2}{\tau} \int_1^{1/\alpha} \frac{1}{s} \E[\|X\|] \d s 
 = \frac{2n}{\tau} \ln \frac{1}{\alpha} \qedhere
\end{align*}
\end{proof}

\section{The shadow simplex method with symbolic perturbation}
\label{sec:shadow-simplex-alg}

In this section, we will give a self-contained presentation of the algorithmic details of the shadow simplex method,
including the details of coping with degenerate $P$ using a perturbation of the right-hand sides $\vec{b}$.
For clarity of presentation, we will first consider the case of simple polyhedra.
We will also assume that $\vec{c}$ and $\vec{d}$ are in general position as made precise in the precondition of Algorithm~\ref{alg:shadow-simplex}.

\begin{algorithm}[]
  \DontPrintSemicolon
  \KwIn{$P = \set{ \vec{x} \in \R^n : A\vec{x} \leq \vec{b} }$,
    $\vec{c}, \vec{d} \in N(P)$,
    optimal basis $B \subset [m]$ for $\vec{c}$}
  \KwReq{$P$ is pointed and simple}
  \KwReq{$[\vec{c}, \vec{d})$ intersects facets of normal cones only in their relative interiors}
  \KwOut{optimal basis $B \subset [m]$ for $\vec{d}$}
  $\lambda \gets 0$\;
  Gauss elimination: $A \gets A U$, $\vec{c} \gets U^T \vec{c}$, $\vec{d} \gets U^T \vec{d}$ so that $A_B = I_n$\;
  \Loop{}{
    $i^\star \gets \arg\min\set{ \frac{\vec{c}_i}{\vec{c}_i - \vec{d}_i} : i \in
B, \vec{c}_i > \vec{d}_i }$\;
    \lIf{$i^\star$ undefined or $\lambda^\star = \frac{\vec{c}_{i^\star}}{\vec{c}_{i^\star} - \vec{d}_{i^\star}} \geq 1$}{
      \Return{B}
    }
    $j^\star \gets \arg\min\set{ \frac{\pr{\vec{a}_j}{\vec{b}_B} - \vec{b}_j}{\vec{a}_{ji^\star}} : j \not\in B, \vec{a}_{ji^\star} < 0}$\;
    $B \gets B \setminus \set{ B_{i^\star} } \cup \set{ j^\star }$, $\lambda \gets \lambda^\star$\;
    Gauss elimination: $A \gets A U$, $\vec{c} \gets U^T \vec{c}$, $\vec{d} \gets U^T \vec{d}$ so that $A_B = I_n$\;
  }
  \caption{Shadow Simplex\label{alg:shadow-simplex}}
\end{algorithm}

\begin{lemma}
  \label{lemma:simplex-col-elim}
  Let $A \in \R^{m \times n}$, $\vec{b} \in \R^m$, and $U \in \R^{n \times n}$ invertible.
  Then a basis $B$ is optimal for $\max\set{ \pr{\vec{c}}{\vec{x}} : A\vec{x} \leq \vec{b}}$
  if and only if it is optimal for $\max\set{ \pr{U^T \vec{c}}{\vec{x}} : AU\vec{x} \leq \vec{b}}$.
\end{lemma}
\begin{proof}
  Let $\vec{a}_1, \dots, \vec{a}_m \in \R^n$ be the rows of $A$.
  The basis $B$ is optimal for the first problem if and only if $\vec{c} \in \cone\set{ \vec{a}_i : i \in B }$.
  This is equivalent to $U^T \vec{c} \in \cone\set{ U^T \vec{a}_i : i \in B }$.
  Since the $U^T \vec{a}_i$ are the rows of $A U$, this is equivalent to $B$ being an optimal basis for the second problem.
\end{proof}

\begin{theorem}
  \label{thm:simplex-method-simple}
  Algorithm~\ref{alg:shadow-simplex} is correct as specified and requires $O(mn^2 + Nmn)$ arithmetic operations,
  where $N$ is the number of normal cone facets intersected by $[\vec{a},\vec{b}]$.
\end{theorem}
\begin{proof}
  The initial Gauss elimination requires $O(mn^2)$ arithmetic operations.
  Each iteration is dominated by the computation of $j^\star$ and the rank-$1$ Gauss elimination update,
  both of which require $O(mn)$ arithmetic operations.

  We will show the invariant that $B$ is an optimal basis for $\max\set{ \pr{\vec{c}_\lambda}{\vec{x}} : A\vec{x} \leq \vec{u} }$,
  where $\vec{c}_\lambda = (1-\lambda) \vec{c} + \lambda \vec{d}$.
  The invariant initially holds by definition of the input
  and remains unchanged by the Gauss elimination steps due to Lemma~\ref{lemma:simplex-col-elim}.

  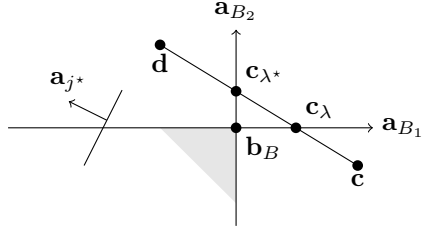
\begin{figure}
    \begin{center}
      \begin{tikzpicture}
        \fill[black!10] (-1,0) -- (0,0) -- (0,-1) -- cycle;
        \draw (-3,0) -- (0,0) -- (0,-1.3);
        \fill (0,0) circle[radius=2pt] node[below right] {$\vec{b}_B$};

        \draw[name path=B1,->] (0,0) -- (1.8,0) node[right] {$\vec{a}_{B_1}$};
        \draw[name path=B2,->] (0,0) -- (0,1.3) node[above] {$\vec{a}_{B_2}$};

        \coordinate (c) at (1.6,-0.5);
        \coordinate (d) at (-1,1.1);

        \draw[name path=cd] (c) -- (d);

        \path[name intersections={of=cd and B1}] (intersection-1) node (cL) {};
        \path[name intersections={of=cd and B2}] (intersection-1) node (cLs) {};

        \foreach \p in {c,d,cL,cLs} {
          \fill (\p) circle[radius=2pt];
        }
        \draw (c) node[below] {$\vec{c}$};
        \draw (d) node[below] {$\vec{d}$};
        \draw (cL) node[above right] {$\vec{c}_\lambda$};
        \draw (cLs) node[above right] {$\vec{c}_{\lambda^\star}$};

        \draw (-2,-0.5) -- (-1.5,0.5);
        \draw[->] (-1.7,0.1) -- +(-0.5,0.25) node[above] {$\vec{a}_{j^\star}$};
      \end{tikzpicture}
    \end{center}
    \caption{One iteration of the Shadow Simplex\label{fig:shadow-simplex-iteration}}
  \end{figure}

  A typical iteration is illustrated in Figure~\ref{fig:shadow-simplex-iteration}.
  The vertex corresponding to $B$ is $\vec{b}_B$, its normal cone is the positive orthant $\R^n_+$.
  The invariant implies that $\vec{c}_{\lambda} \in \R^n_+$.
  If $\vec{d} - \vec{c} \geq 0$, the algorithm returns $B$.
  Indeed, $B$ is an optimal basis for $\vec{d}$ in this case because $\vec{d} \in \R^n_+$.

  Otherwise, the ray from $\vec{c}_\lambda$ through $\vec{d}$ eventually leaves the positive orthant,
  and in fact $\vec{c}_{\lambda^\star}$ is the last point contained in $\R^n_+$.
  If $\lambda^\star \geq 1$ we know that $\vec{d} = \vec{c}_1 \in \R^n_+$
  and so the current basis is optimal for $\vec{d}$.
  The index $i^\star$ is the index into the basis
  whose contribution to the conic combination describing $\vec{c}_{\lambda^\star}$ is $0$.%
  \footnote{This is just a different way of saying that $(\vec{c}_{\lambda^\star})_{i^\star} = 0$.
  Due to the precondition on $[\vec{c},\vec{d})$ this index is uniquely defined when $\lambda^\star < 1$.}

  Letting $B' = B \setminus \{ B_{i^\star} \} \cup \{ j^\star \}$,
  where $B_{i^\star}$ is the $i^\star$-th index in the basis,
  it is trivially true that $\vec{c}_{\lambda^\star} \in \cone A_{B'}^T$,
  so it only remains to show that $B'$ is a feasible basis.
  The edge of the polyhedron described by the constraints $B \setminus \{ B_{i^\star} \}$
  is contained in the ray starting at the vertex $\vec{b}_B$ in direction $-\vec{e}_{i^\star}$.
  This ray can only be cut off by constraints $\vec{a}_j \vec{x} \leq \vec{b}_j$ with $\vec{a}_{ji^\star} < 0$.
  At least one such constraint must exist by the condition that $\vec{d}$ lies in the support of the normal fan,
  i.e. the corresponding linear program is bounded.
  The fraction $\frac{\pr{\vec{a}_j}{\vec{b}_B} - \vec{b}_j}{\vec{a}_{ji^\star}}$ is the Euclidean distance from $\vec{b}_B$
  to the intersection of the constraint $\pr{\vec{a}_j}{\vec{x}} \leq \vec{b}_j$ with the ray,
  and so $B'$ is feasible.\footnote{%
  Note that the minimum is positive and unique because the underlying polyhedron is simple.}
  This completes the proof of the invariant and thus the proof of correctness.

  If the initial objective $\vec{c}$ lies in the facet of a normal cone we may get $\lambda^\star = \lambda$ in the first iteration.
  Nevertheless, due to the precondition on $[\vec{c}, \vec{d})$,
  every computed value $\lambda^\star$ is distinct
  and, except for the last one,
  corresponds to one intersection point of $[\vec{c},\vec{d}]$ with a facet of a normal cone.
  Therefore, the number of iterations is bounded by $N + 1$.
\end{proof}

If the input polyhedron were non-simple, the proof of correctness and termination would still work
given $\vec{c}$ in sufficiently general position.
However, the $\lambda^\star$ would then correspond to points of intersection between $[\vec{c}, \vec{d}]$ and
cones corresponding to bases. Those cones are merely subsets of normal cones,
and they may not be mutually consistent with a single triangulation of the normal fan.
For this reason, we consider a perturbed polyhedron
\[
  P_\varepsilon := \set{ \vec{x} \in \R^n : A\vec{x} \leq \vec{b} + \vec{\gamma}(\varepsilon) }
\]
where $\vec{\gamma}(\varepsilon) := (\varepsilon, \varepsilon^2, \ldots, \varepsilon^m)$ for $\varepsilon > 0$.

\begin{lemma}
  \label{lemma:perturbation}
  If $\varepsilon > 0$ is sufficiently small, one has that
  \begin{enumerate}
    \item every feasible basis $B$ for $P_\varepsilon$ is also feasible for $P$,
    \item the normal fan of $P_\varepsilon$ triangulates the normal fan of $P$, and
    \item if $B = \{ m - n + 1, \ldots, m \}$ is a feasible basis for $P$, then it is also feasible for $P_\varepsilon$.
  \end{enumerate}
\end{lemma}
\begin{proof}
  If $P_\varepsilon$ is not simple,
  there is some basis $B$ and a constraint $j \not\in B$ such that
  \[
    \pr{ \vec{a}_j }{ A_B^{-1} (\vec{b}_B + \vec{\gamma}(\varepsilon)_B) } = \vec{b}_j + \varepsilon^j
  \]
  In other words, $\varepsilon$ is a root of a non-zero polynomial.
  Since there are only finitely many pairs of $B$ and $j \not\in B$,
  we have that $P_\varepsilon$ is simple for sufficiently small $\varepsilon > 0$.
  In particular, the normal fan of $P_\varepsilon$ is a triangulation.
  For the second claim, it remains to show that for all vertices $\vec{x} \in P_\varepsilon$
  one has $N_\vec{x} \subseteq N_\vec{y}$ for some vertex $\vec{y} \in P$.
  In fact, this is implied by the first claim, which we show next.

  Let $B$ be an infeasible basis for $P$,
  i.e. there is some $j \not\in B$ such that
  \[
    \pr{ \vec{a}_j }{ A_B^{-1} (\vec{b}_B + \vec{\gamma}(\varepsilon)_B) } > \vec{b}_j + \varepsilon^j
  \]
  holds for $\varepsilon = 0$.
  Since both sides of the inequality are continuous as functions in $\varepsilon$,
  strict inequality also holds for all sufficiently small $\varepsilon > 0$.
  There are only finitely many infeasible bases for $P$,
  so we have that all of them are infeasible for $P_\varepsilon$ when $\varepsilon > 0$ is sufficiently small.
  This implies the first claim.

  Finally, if $B = \{ m - n + 1, \ldots, m \}$ is feasible for $P$,
  we have that
  \[
    \pr{ \vec{a}_j }{ A_B^{-1} (\vec{b}_B + \vec{\gamma}(\varepsilon)_B) } - \vec{b}_j - \varepsilon^j \leq 0
  \]
  holds for all $j \leq m - n$ when $\varepsilon = 0$.
  The left hand side is a polynomial in $\varepsilon$ whose lowest-degree non-constant monomial is $-\varepsilon^j$.
  This implies that the inequality also holds when $\varepsilon > 0$ is sufficiently small,
  hence $B$ is feasible for $P_\varepsilon$.
\end{proof}

Without explicit bounds on the coefficients describing $P$, we cannot give a quantitative bound for $\varepsilon$.
We avoid the need for such a bound by applying the perturbation symbolically.
Since the right-hand sides $\vec{b}$ never appear in divisors,
we can perform related computations in the polynomial ring $R := \R[\varepsilon]$.
The order $\leq$ on $\R$ naturally extends to a lexicographic order on $R$ such that $a \leq b$ holds for $a, b \in R$
if and only if $a(\varepsilon) \leq b(\varepsilon)$ holds over the reals for all sufficiently small $\varepsilon > 0$.

\begin{theorem}
  \label{thm:shadow-simplex}
  Let $P = \set{ \vec{x} \in \R^n : A\vec{x} \leq \vec{b} }$ be pointed,
  $\vec{c}, \vec{d} \in \R^n$, and $B$ an optimal basis for $\vec{c}$.
  If every intersection of $[\vec{c}, \vec{d})$ with a facet $F$ of a cone spanned by a feasible basis of $P$
  lies in the relative interior of $F$,
  the Shadow Simplex can be used to compute an optimal basis for $\vec{d}$
  in $O(mn^2 + Nmn)$ arithmetic operations,
  where $N$ is the number of intersections of $[\vec{c},\vec{d}]$
  with some triangulation of the normal fan of $P$
  that contains the cone spanned by the initial basis $B$.
\end{theorem}
\begin{proof}
  Rearrange the rows of $A\vec{x} \leq \vec{b}$ so that $B = \set{ m - n + 1, \ldots, m }$
  and apply Lemma~\ref{lemma:perturbation}.
  This gives us all preconditions of the Shadow Simplex algorithm,
  and the proof of Theorem~\ref{thm:simplex-method-simple} applies
  with a single caveat:
  the computation of $j^\star$ involves computations and comparisons of terms of the form
  $\frac{\pr{\vec{a}_j}{\vec{b}_B + \vec{\gamma}(\varepsilon)_B} - \vec{b}_j - \varepsilon^j}{\vec{a}_{ji^\star}} \in R$.
  The resulting polynomials contain at most $n+2$ monomials.
  By storing them as sparse sorted vectors,
  we can compute each term and compare it to the previous best in time $O(n)$,
  so that the computation of $j^\star$ still requires only $O(mn)$ arithmetic operations.
\end{proof}

\section{Removing unboundedness}
\label{sec:removing-unboundedness}

In this section, we discuss two approaches to add constraints to a pointed polyhedron
$P = \set{ \vec{x} \in \R^n : A\vec{x} \leq \vec{b} }$ to construct an LP equivalent
\emph{polytope} $P'$.
This requires the additional constraints to be strictly valid for all vertices of $P$.

\begin{lemma}
  \label{lem:delta-distance-inverse}
  Let $M \in \R^{n \times n}$ be a matrix whose rows are
  $\vec{v}_1, \ldots, \vec{v}_n \in \bS^{n-1}$.
  Then $\vec{v}_1, \ldots, \vec{v}_n$ satisfy the $\delta$-distance property
  if and only if the columns $\vec{u}_1,\ldots,\vec{u}_n$ of $M^{-1}$ satisfy $\|\vec{u}_j\| \leq 1/\delta$ for all $j$.
\end{lemma}
\begin{proof}
  Let $H_i = \linsp\set{ \vec{v}_j : j \neq i }$.
  Then
  \[
    1 = \pr{\vec{v}_i}{\vec{u}_i} = \|\vec{u}_i\| d(\vec{v}_i, H_i) \iff d(\vec{v}_i, H_i) = \frac{1}{\|\vec{u}_i\|}
  \]
  implies the statement.
\end{proof}

\begin{lemma}
  \label{lem:vertex-norm}
  Let $P = \set{ \vec{x} \in \R^n : A\vec{x} \leq \vec{b} }$
  satisfy the local $\delta$-distance property
  and let $b_{\max} = \max\set{ \frac{|\vec{b}_i|}{\|\vec{a}_i\|} : i \in [m] }$.
  Then every vertex $\vec{x} \in P$ satisfies $\|\vec{x}\| \leq \frac{n b_{\max}}{\delta}$.
\end{lemma}
\begin{proof}
  We may assume without loss of generality that $\|\vec{a}_i\| = 1$ for all $i \in [m]$.
  Let $\vec{x}$ be a vertex and let $B$ be a basis for $\vec{x}$.
  Let $\vec{u}_1, \ldots, \vec{u}_n \in \R^n$ be the columns of $A_B^{-1}$.
  The triangle inequality and Lemma~\ref{lem:delta-distance-inverse} imply
  $\|\vec{x}\| = \|A_B^{-1} \vec{b}_B\| \leq \sum_{i \in B} \|\vec{u}_i\| b_{\max} \leq n b_{\max} / \delta$.
\end{proof}

\begin{proof}[Proof of Lemma~\ref{lem:unb-to-bounded}]
\begin{enumerate}
\item Let $P = \set{\vec{x} \in \R^n : A\vec{x} \leq \vec{b}}$ satisfy the local $\delta$-distance property.
  For some feasible basis $I$ we let
  \[
    \vec{w} = - \frac{1}{n} \sum_{i \in I} \frac{\vec{a}_i}{\|\vec{a}_i\|}
  \]
  and $P' = \set{ \vec{x} \in P : \pr{\vec{w}}{\vec{x}} \leq \frac{n b_{\max}}{\delta} }$.
  The normal fan of $P'$ covers $\R^n$,
  because $\vec{0}$ lies in the interior of $\conv(\set{\vec{a}_i : i \in I} \cup \set{\vec{w}})$,
  and so $P'$ is a polytope.
  Since $\|\vec{w}\| < 1$, we have $\pr{\vec{w}}{\vec{x}} < \|\vec{x}\| \leq \frac{n b_{\max}}{\delta}$ for every vertex $\vec{x} \in P$
  by Lemma~\ref{lem:vertex-norm}, so $P'$ is LP equivalent to $P$.

  Every vertex of $P'$ is either a vertex of $P$ or the intersection of an unbounded ray of $P$ with the new constraint.
  Consider a feasible basis $B$ of $P'$.
  Either $B$ is already feasible for $P$, in which case it satisfies the $\delta$-distance property.
  Otherwise, it is of the form
  $\vec{a}_1, \ldots, \vec{a}_{n-1}, \vec{w}$,
  where $\vec{a}_1, \ldots, \vec{a}_n$ is a feasible basis for $P$,
  after a suitable renumbering of indices.

  In this case, $\cone\set{\vec{a}_1, \ldots, \vec{a}_{n-1}}$ is on the boundary of the support of the normal fan of $P$,
  and so by the proof of Lemma~\ref{lem:dist-to-width} we have $d(\vec{w}, H) \geq \delta / n$,
  where $H = \linsp\set{\vec{a}_1, \ldots, \vec{a}_{n-1}}$.
  Using an orthogonal transformation,
  we may assume without loss of generality that $\vec{a}_{1n} = \dots = \vec{a}_{(n-1),n} = 0$
  and so the matrix $M$ whose rows are the basis vectors normalized to unit length is of the form
  \[
    M = \begin{pmatrix}
          A'         & 0 \\
          \vec{w}'^T & h
        \end{pmatrix} \in \R^{n \times n}
  \]
  where $h = d(\vec{w}/\|\vec{w}\|, H) \geq \delta / n$ and $\|\vec{w}'\| < 1$.
  We compute
  \[
    M^{-1} = \begin{pmatrix}
                A'^{-1}                 & 0 \\
                -\vec{w}'^T A'^{-1} / h & 1/h
             \end{pmatrix}
  \]
  By Lemma~\ref{lem:delta-distance-inverse},
  it is sufficient to show that the norms of the columns of $M^{-1}$ are bounded by $2n/\delta^2$.
  This is immediate for the last column.
  Let $\vec{u}_1, \ldots, \vec{u}_{n-1}$ be the columns of $A'^{-1}$.
  We have $\|\vec{u}_i\| \leq 1/\delta$ by Lemma~\ref{lem:delta-distance-inverse}.
  Furthermore, the $i$-th entry of the last row of $M^{-1}$ is bounded in absolute value by
  \[
    \left| \frac{\pr{\vec{w}'}{\vec{u}_i}}{h} \right| \leq \frac{\|\vec{u}_i\| n}{\delta} \leq \frac{n}{\delta^2}
  \]
  By the triangle inequality, the norms of the first $n-1$ columns of $M^{-1}$ are bounded by $1/\delta + n/\delta^2 \leq 2n/\delta^2$.
  This completes the proof that $P'$ satisfies the local $\delta^2/(2n)$-distance property.

\item
  Now suppose that $A$ satisfies the global $\delta$-distance property.
  For every vertex $\vec{x} \in P$, we have
  \[
    \pr{-\vec{a}_i}{\vec{x}} \leq \|\vec{a}_i\| \|\vec{x}\| \leq \frac{n \|\vec{a}_i\| b_{\max}}{\delta}
  \]
  by Lemma~\ref{lem:vertex-norm}, and so the polytope
  \[
    P' = \set{\vec{x} \in \R^n: -\frac{n\|\vec{a}_i\|b_{\rm max}}{\delta} - 1 \leq \pr{\vec{a}_i}{\vec{x}} \leq \vec{b}_i \quad \forall i \in [m]}
  \]
  is LP equivalent to $P$.
  Furthermore, every (not necessarily feasible) basis of the constraint matrix of $P'$ is equal to a basis of $A$ up to sign changes,
  which do not affect the $\delta$-distance property.
  Therefore, $P'$ satisfies the global $\delta$-distance property.
  \qedhere
\end{enumerate}
\end{proof}

\section{Feasibility}
\label{sec:feasibility}

\begin{lemma}
  Let $P = \set{ \vec{x} \in \R^n : A \vec{x} \leq \vec{b} }$
  where $A \in \Z^{m \times n}$ is an integral matrix of full column rank with subdeterminants bounded by $\Delta$.
  One can compute a feasible basis of $P$ or decide infeasibility
  using an expected $O(n^5 \Delta^2 \ln n\Delta)$ shadow simplex pivots.
\end{lemma}
\begin{proof}
  Consider the linear program
  \begin{align*}
    \min\, s &\\
    \pr{\vec{a}_i}{\vec{x}} - s &\leq \vec{b}_i & \forall i \in [m] \\
                              s &\geq 0
  \end{align*}
  The $(m + 1) \times (n + 1)$-constraint matrix is integral of full column rank and has $n \times n$-subdeterminants bounded by $n\Delta$.
  Therefore, it satisfies the global $\delta$-distance property with $\delta = \frac{1}{n^2\Delta^2}$ by Lemma~\ref{lem:determinant-delta-dist}.
  The point $(\vec{0}, -\min(\set{0} \cup \set{ \vec{b}_i : i \in [m]}))$ is feasible,
  and so a feasible basis can be found using a standard ray-casting procedure.
  Lemma~\ref{lem:unb-to-bounded} implies that we can construct an LP equivalent polytope with the same parameter $\delta$,
  which we can then optimize in $O(n^5 \Delta^2 \ln n\Delta)$ shadow simplex pivots
  by Theorem~\ref{thm:phase2}.
  If the optimal solution we found satisfies $s = 0$, we can read off a feasible basis of $P$;
  otherwise, we know that $P$ is empty.
\end{proof}

\begin{lemma}
  Let $P = \set{ \vec{x} \in \R^n : A \vec{x} \leq \vec{b} }$
  where $A \in \R^{m \times n}$ satisfies the global $\delta$-distance property.
  One can compute a feasible basis of $P$ or decide infeasibility
  using an expected $O(mn^3/\delta \ln n/\delta)$ shadow simplex pivots.
\end{lemma}
\begin{proof}
  We proceed by iteratively adding the constraints of $P$, one at a time.
  First, observe that we can find a basis $B$ of $A$ efficiently using Gauss elimination,
  which gives us the unique feasible basis of $P_B = \set{ \vec{x} \in \R^n : A_B \vec{x} \leq \vec{b}_B }$.

  Now suppose we already found a feasible basis $B$ of $P_I = \set{ \vec{x} \in \R^n : A_I \vec{x} \leq \vec{b}_I }$, $I \subsetneq [m]$.
  Let $i \not\in I$.
  Since $P_I$ satisfies the global $\delta$-distance property,
  we can combine Lemma~\ref{lem:unb-to-bounded} with Theorem~\ref{thm:phase2}
  to solve the linear program
  \[
    \min\set{ \pr{\vec{a}_i}{\vec{x}} : \vec{x} \in P_I } = \gamma
  \]
  using an expected $O(n^3/\delta \ln n/\delta)$ shadow simplex pivots.
  If $\gamma > \vec{b}_i$, this implies that $P_{I + i}$ is empty and therefore $P$ is empty.
  Otherwise, the solution of the linear program yields a point $\vec{x} \in P_{I + i}$
  (if $\gamma = -\infty$, we find a suitable point on an unbounded ray)
  which we can round to a feasible basis of $P_{I + i}$ using a standard ray-casting procedure if necessary.

  Applying this procedure iteratively until $I = [m]$,
  we use an expected $O(mn^3/\delta \ln n/\delta)$ shadow simplex pivots to obtain a feasible basis of $P$.
  Observe that the time required for those pivots dominates the time required for any intermediate ray-casting.
\end{proof}

\section{The perfect matching polytope}
\label{sec:perfect-matching}

The perfect matching polytope $P_G \subset \R^E$ of an undirected graph $G = (V, E)$, $|V| = 2n$,
is the convex hull of the characteristic vectors $\chi_M$ of perfect matchings $M \subseteq E$
(see chapter~25 of \cite{book:Schrijver2003} for a collection of fundamental results on $P_G$).
It is described by the system of inequalities
\begin{align*}
  \sum_{e \in \delta(v)} x(e) & \leq 1 & \forall v \in V \\
  \sum_{e \in \delta(U)} x(e) & \geq 1 & \forall U \subseteq V, |U| \text{ odd} \\
                         x(e) & \geq 0 & \forall e \in E
\end{align*}
We will also use the fact that two vertices $\chi_M$ and $\chi_N$ of $P_G$
are adjacent if and only if $M \triangle N$ is a cycle.

We will show that even though this polytope has an exponential number of facets,
so that its normal fan contains an exponential number of extreme rays,
every normal cone is rather wide.
This is due to the high level of degeneracy of the polytope.
While a much better bound on the diameter of $P_G$ follows directly from
the combinatorial adjacency structure noted above,
it is interesting to see that some of our techniques can be applied to $P_G$.
As far as we know, this is the first example of a combinatorial polytope
that satisfies this kind of ``discrete curvature bound'' without having a constraint matrix with small subdeterminants.

Since $P_G$ is not full-dimensional,
there are two different but essentially equivalent definitions for the normal cones of $P_G$.
One can treat $P_G$ as a polytope in the ambient space $\R^E$,
keeping the definition of the normal cone $N_v$ of a vertex $\vec{v} \in P_G$
as the set of objective functions $\vec{c} \in \R^E$ that are maximized at $\vec{v}$.
Hence $N_v$ is not pointed,
and its lineality space $L^\bot$ is the set of vectors that are orthogonal to the affine hull $\aff(P_G)$ of $P_G$.
Alternatively, one may treat $P_G$ as a full-dimensional polytope within $\aff(P_G)$,
in which case the normal cone $N_v'$ is simply the restriction of $N_v$
to the linear space $L$ of vectors parallel to $\aff(P_G)$.
The choice of definition does not affect $\tau$-width:
given a ball $\vec{w} + \tau B_2^{E} \subseteq N_v$, $\|\vec{w}\| = 1$,
the orthogonal projection of the ball onto $L$ is $\vec{w}' + \tau B_2^L \subseteq N_v'$
with $\|\vec{w}'\| \leq 1$.

\begin{theorem}
  $P_G$ is $\tau$-wide for $\tau = 1/(3\sqrt{|E|})$.
\end{theorem}
\begin{figure}
  \begin{center}
    \begin{tikzpicture}
      \foreach \x in {0,1,2,6} {
        \fill (\x,0) circle[radius=2pt];
        \fill (\x,1.5) circle[radius=2pt];
        \draw[very thick] (\x,0) -- (\x,1.5);
      }
      \draw (7,0) -- (7,1.5);
      \draw[fill=white] (7,0) circle[radius=2pt];
      \draw[fill=white] (7,1.5) circle[radius=2pt];

      \foreach \x in {0,1,2} {
        \draw (\x,-0.2) node[below] {$v_{\x}$};
        \draw (\x, 1.7) node[above] {$u_{\x}$};
      }
      \draw (6,-0.2) node[below] {$v_{n-1}$};
      \draw (6, 1.7) node[above] {$u_{n-1}$};
      \draw (7,-0.2) node[below right] {$v_n = v_0$};
      \draw (7, 1.7) node[above right] {$u_n = u_0$};

      \draw (4,0.75) node {$\dots$};

      \foreach \x in {0,1,2,6} {
        \draw (\x,1.5) ++(30:0.2)
              arc[start angle=30,end angle=180,radius=0.20cm] -- ++(0,-1.5)
              arc[start angle=180,end angle=270,radius=0.20cm] -- ++(1,0.1)
              arc[start angle=270,end angle=390,radius=0.10cm] -- cycle;

        \draw (\x,0) ++(330:0.15)
              arc[start angle=330,end angle=180,radius=0.15cm] -- ++(0,1.5)
              arc[start angle=180,end angle=90,radius=0.15cm] -- ++(1,-0.05)
              arc[start angle=90,end angle=-30,radius=0.10cm] -- cycle;
      }
    \end{tikzpicture}
  \end{center}
  \caption{A perfect matching and a selection of tight odd sets.\label{fig:perfect-matching-odd-sets}}
\end{figure}
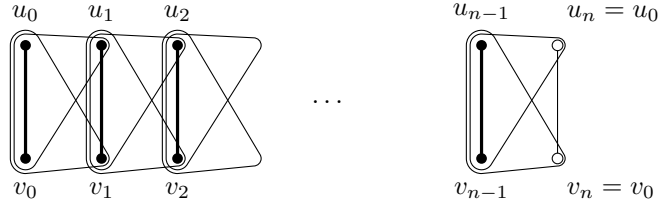
\begin{proof}
  Let $\chi_M \in P_G$ be a vertex and let $N_{\chi_M}$ be its normal cone.
  Let us label the vertices of $G$ such that
  \[
    M = \{ u_0v_0, u_1v_1, \ldots, u_{n-1}v_{n-1} \}
  \]
  For an odd set $U \subseteq V$ of vertices,
  let $\vec{a}_U$ be the corresponding row of the constraint matrix in normal form,
  i.e. $\vec{a}_U = - \chi_{\delta(U)}$.
  We consider the following $3$-element sets, see Figure~\ref{fig:perfect-matching-odd-sets}:
  \[
    \cU := \set{ \set{ u_k, v_k, w } : 0 \leq k \leq n -1, w \in \set{ u_{k + 1\pmod{n}}, v_{k + 1\pmod{n}} } }
  \]
  Since the corresponding constraints are tight at $\chi_M$,
  we have that
  \[
    \vec{w} := \sum_{U \in \cU} \vec{a}_U \in N_{\chi_M}
  \]
  Note that for $e \in M$, we have $\vec{w}(e) = -2$,
  and for $e \in E \setminus M$, we have $\vec{w}(e) \in \{ -4, -6 \}$
  (every vertex is contained in exactly $3$ of the sets in $\cU$,
  and there can be at most one set in $\cU$ that contains both endpoints of $e \not\in M$).
  In particular, $\|\vec{w}\| \leq 6 \sqrt{|E|}$.

  Every facet $F$ of $N_{\chi_M}$ corresponds to an edge from $\chi_M$ to some other vertex $\chi_N$.
  We know that $M \triangle N$ is a cycle $C$.
  The direction $\vec{v}$ of the edge from $\chi_N$ to $\chi_M$, which is (inner) normal to $F$,
  satisfies
  \[
    \vec{v}(e) = \begin{cases}
                     1 & e \in C \cap M \\
                    -1 & e \in C \setminus M \\
                     0 & e \not\in C
                 \end{cases}
  \]
  We compute
  \begin{align*}
    \pr{\vec{v}}{\vec{w}}
      &= \sum_{e \in C \cap M} \vec{w}(e) - \sum_{e \in C \setminus M} \vec{w}(e)
      \geq -2 |C \cap M| + 4 |C \setminus M| = |C|
  \end{align*}
  where we use the fact that $C$ alternates between edges in $M$ and edges in $N$ for the last equation.
  Let $H$ be the affine span of $F$.
  We obtain
  \[
    d(\vec{w}, H) = \frac{\pr{\vec{v}}{\vec{w}}}{\|\vec{v}\|} \geq \frac{|C|}{\|\vec{v}\|} = \sqrt{|C|} \geq 2
  \]
  In other words, $N_{\chi_M}$ contains a ball of radius $2$ around $w$,
  and so $N_{\chi_M}$ is $\tau$-wide for $\tau = 2/\|\vec{w}\| \geq 1/(3\sqrt{|E|})$.
\end{proof}

\subparagraph*{Acknowledgements}

We would like to thank Friedrich Eisenbrand and Santosh Vempala for useful
discussions.




\bibliographystyle{plain}
\bibliography{references}

\section{Additional proofs}
\label{sec:appendix-proofs}

\begin{proof}[Proof of Lemma~\ref{lem:dist-to-width}]
Let $\vec{v}_1^*,\dots,\vec{v}_n^*$ be the dual basis satisfying
$\pr{\vec{v}_i}{\vec{v}_j^*} = 1$ if $i=j$ and $0$ otherwise. By definition
of the $\delta$-distance property
\[
\pr{\vec{v}_i}{\vec{v}_i^*/\|\vec{v}_i^*\|} \geq \delta
\Rightarrow \|\vec{v}_i^*\| \leq 1/\delta \text{.}
\]
Note that $\vec{x} \in \cone(\vec{v}_1,\dots,\vec{v}_n)$ iff
$\pr{\vec{v}_i^*}{\vec{x}} \geq 0$, for all $i \in [n]$. Let $\bar{\vec{v}} =
\sum_{i=1}^n \vec{v}_i/n$, and note that $\|\bar{\vec{v}}\| \leq 1$, since
it is an average of unit vectors.

We will show that $\bar{\vec{v}} + \frac{\delta}{n} \cB_2^n \subseteq
\cone(\vec{v}_1,\dots,\vec{v}_n)$, which suffices to prove the lemma. Take any
vector $\vec{e}$, $\|\vec{e}\| \leq \delta/n$. Then for any $i \in [m]$, note
that
\[
\pr{\vec{v}_i^*}{\bar{\vec{v}}+\vec{e}} = 1/n + \pr{\vec{v}_i^*}{\vec{e}}
\geq 1/n - \|\vec{v}_i^*\|\|\vec{e}\| \geq 0 \text{.}
\]
Hence $\bar{\vec{v}} + \vec{e} \in \cone(\vec{v}_1,\dots,\vec{v}_n)$, as needed.
\end{proof}

\begin{proof}[Proof of Lemma~\ref{lem:exp-moments}]
For the first part,
\begin{align*}
c_\Sigma^{-1} &= \int_{\Sigma} e^{-\|\vec{x}\|} \d\vec{x}
    = \int_{\Sigma} \int_{\|\vec{x}\|}^\infty e^{-t} {\rm dt} \d\vec{x} \\
    &= \int_0^\infty e^{-t} \int_{\Sigma} {\rm I}[\|\vec{x}\| \leq t] \d\vec{x}{\rm dt}
    = \int_0^\infty e^{-t} t^n \vol_n(\cB_2^n \cap \Sigma){\rm dt} = n!
\vol_n(\cB_2^n \cap \Sigma) \text{.}
\end{align*}
For the expected norm,
\begin{align*}
  \E[\|X\| ]
    &= c_\Sigma \int_{\Sigma} \|\vec{x}\| e^{-\|\vec{x}\|} \d\vec{x}
     = c_\Sigma \int_{\Sigma} \|\vec{x}\| \int_{\|\vec{x}\|}^\infty e^{-t} {\rm dt} \d\vec{x} \\
    &= c_\Sigma \int_0^\infty  e^{-t} \int_{t(\cB_2^n \cap \Sigma)} \|\vec{x}\|
        \d\vec{x} {\rm dt}
     = c_\Sigma \int_0^\infty  e^{-t} t^{n+1} {\rm dt}
          \int_{\cB_2^n \cap \Sigma} \|\vec{x}\| \d\vec{x} \\
    &= c_\Sigma (n+1)! \int_0^1 (1-s^n) \vol_n(\cB_2^n \cap \Sigma) {\rm ds} \\
 &= c_\Sigma (n+1)! \vol_n(\cB_2^n \cap \Sigma) \frac{n}{n+1} = n \text{.} \qedhere
\end{align*}
\end{proof}

\begin{proof}[Proof of Lemma~\ref{lem:delta-dist-equiv}]
\begin{description}
  \item[$(1) \Rightarrow (2)$]
    Let $I \subseteq [m]$ for which $\set{\vec{x}_i: i \in I}$ are linearly
    independent, and examine the linear combination $\sum_{i \in I} a_i
    \vec{x}_i$. Letting $j = \argmax_{i \in I} |a_i|$, we have that
    \begin{align*}
    \|a_j\vec{x}_j + \sum_{i \in I \setminus \set{j}} a_i\vec{x}_i\|
      &\geq d(a_j\vec{x}_j,~\linsp(\set{\vec{x}_i: i \in I \setminus \set{j}})) \\
      &= |a_j| d(\vec{x}_j,~\linsp(\set{\vec{x}_i: i \in I \setminus \set{j}}))
      \geq \delta |a_j| \text{,}\quad \left(\text{ by property $(1)$ }\right)
    \end{align*}
    as needed.

  \item[$(2) \Rightarrow (1)$]
    Take $i \in [m]$ and $J \subseteq [m]$,
    such that $\vec{x}_i \not\in \linsp(\set{\vec{x}_j: j \in J})$. Since we need
    only prove a lower bound on $d(\vec{x}_i, \linsp(\set{\vec{x}_j: j \in J}))$, we
    may clearly assume that $\set{\vec{x}_j: j \in \set{i} \cup J}$ are linearly
    independent. Given this, we have that
    \begin{align*}
    d(\vec{x}_i, \linsp(\set{\vec{x}_j: j \in J}))
      &= \min \set{\vec{x}_i - \sum_{j \in J} a_j \vec{x}_j:  a_j \in \R, j \in J} \\
      &\geq \min \set{\delta \max \set{1, \max_{j \in J} |a_j|}: a_j \in \R, j \in J}
      \quad \left(\text{ by property $(2)$ }\right) \\
      &\geq \delta \text{, \quad as needed.} \qedhere
    \end{align*}
\end{description}
\end{proof}

\begin{proof}[Proof of Lemma~\ref{lem:delta-distance-projection}]
  Let $i \in \{1, \dots, k-1\}$.
  Let $S_i = \linsp\set{ \pi(\vec{v}_j) : j \neq i, k }$.
  First, observe
  \[
    \linsp(\set{ \vec{v}_k } \cup S_i) = \linsp(\set{ \vec{v}_j : j \neq i }).
  \]
  Assume that $d(\pi(\vec{v}_i), S_i) < \delta \|\pi(\vec{v}_i)\|$.
  So there exists some $\vec{x} \in S_i$ such that $d(\pi(\vec{v}_i), \vec{x}) < \delta \|\pi(\vec{v}_i)\|$.
  Since we can write $\vec{v}_i = \pi(\vec{v}_i) + \lambda \vec{v}_k$ for some $\lambda \in \R$,
  we get
  \[
    d(\vec{v}_i, \linsp(\set{ \vec{v}_j : j \neq i}))
      \leq d(\vec{v}_i, \vec{x} + \lambda \vec{v}_k) = d(\pi(\vec{v}_i), \vec{x}) < \delta \|\pi(\vec{v}_i)\|
      \leq \delta \|\vec{v}_i\|,
  \]
  which contradicts the $\delta$-distance property of $\vec{v}_1, \ldots, \vec{v}_k$.
  So we must in fact have $d(\pi(\vec{v}_i), S_i) \geq \delta \|\pi(\vec{v}_i)\|$ for all $i \in \{1, \ldots, k-1\}$,
  which completes the proof.
\end{proof}

The following result is already implicit in~\cite{BSEHN14}. We provide it here for completeness.

\begin{lemma}
  \label{lem:determinant-delta-dist}
  Let $A \in \R^{m \times n}$ be an integral matrix whose entries are bounded by $\Delta_1$
  and whose $(n-1)\times(n-1)$ subdeterminants are bounded by $\Delta_{n-1}$ in absolute value.
  Then $A$ satisfies the global $\delta$-distance property with $\delta = 1/(n\Delta_1\Delta_{n-1})$.
  Furthermore, any polyhedron with constraint matrix $A$ is $\tau$-wide with $\tau = 1/(n^2\Delta_1\Delta_{n-1})$.
\end{lemma}
\begin{proof}
  It is sufficient to consider the case where $A \in \Z^{n \times n}$ is invertible.
  Let $\vec{a}_1, \ldots, \vec{a}_n$ be the rows of $A$ and
  let $H_i = \linsp\set{ \vec{a}_j : j \neq i }$.
  The vector $\vec{u}_i$ satisfying $A\vec{u}_i = |\det(A)| \vec{e}_i$ is a normal vector of $H_i$
  with $\vec{u}_i \in \Z^n$ and $\|\vec{u}_i\|_\infty \leq \Delta_{n-1}$ by Cramer's rule.
  We can compute
  \[
    d(\vec{a}_i/\|\vec{a}_i\|, H_i) = \frac{\pr{\vec{a}_i}{\vec{u}_i}}{\|\vec{a}_i\|\|\vec{u}_i\|}
      \geq \frac{1}{(\sqrt{n} \Delta_1) (\sqrt{n} \Delta_{n-1})}
      = \frac{1}{n \Delta_1 \Delta_{n-1}}
  \]
  using the fact that $\|\vec{a}_i\|_\infty \leq \Delta_1$.

  The ``furthermore'' part of the statement of the Lemma follows from Lemma~\ref{lem:dist-to-width}.
\end{proof}

\end{document}